\RequirePackage{amsmath}
\documentclass[12pt,runningheads]{llncs}

\usepackage[usenames]{color}
\usepackage{amssymb}
\usepackage{amsmath}
\usepackage{amsfonts}
\usepackage{amscd}
\usepackage{graphicx}
\usepackage{subfigure}

\usepackage{todonotes}

\usepackage[colorlinks=true,
linkcolor=webgreen,
filecolor=webbrown,
citecolor=webgreen]{hyperref}
\usepackage[capitalize,nameinlink,sort]{cleveref}

\definecolor{webgreen}{rgb}{0,.5,0}
\definecolor{webbrown}{rgb}{.6,0,0}

\usepackage{color}
\usepackage{fullpage}
\usepackage{float}

\usepackage{graphics}
\usepackage{latexsym}
\usepackage{epsf}

\usepackage{mathtools}
\usepackage{bbm}
\usepackage{enumerate}
\usepackage{pgf,tikz,tikz-cd}
\usetikzlibrary{arrows,arrows.meta,automata, calc, chains, decorations, shapes,positioning,decorations.pathmorphing}

\newcommand{\tO}{\mathtt{0}}
\newcommand{\tL}{\mathtt{1}}

\newcommand{\cP}{\mathcal{P}}

\renewcommand{\phi}{\varphi}

\newcommand{\eps}{\varepsilon}

\newcommand{\fract}[1]{\left\{ #1 \right\}}

\newcommand{\nzero}{n^{(0)}}
\newcommand{\none}{n^{(1)}}
\newcommand{\ndelta}{n^{(a)}}
\newcommand{\append}[2]{#1^{(#2)}}

\newcommand{\caseI}[1]{\medskip\noindent\textbf{Case #1:}}
\newcommand{\caseII}[1]{\textbf{Case #1:}}

\begin{document}
\title{State Complexity of Shifts of the Fibonacci Word}

\author{Delaram Moradi\inst{1}\thanks{Research funded by NSERC grant 2024-03725.}\and
Pierre Popoli\inst{1}
\and Jeffrey Shallit\inst{1}$^{*}$
\and
Ingrid Vukusic\inst{2}\thanks{Research funded by the Austrian Science Fund (FWF) 10.55776/J4850.}}

\authorrunning{D. Moradi et al.}

\institute{
School of Computer Science, University of Waterloo,\\
200 University Ave.\ W., Waterloo, ON N2L 3G1, Canada\\
\href{mailto:delaram.moradi@uwaterloo.ca}{\tt delaram.moradi@uwaterloo.ca}\\
\href{mailto:pierre.popoli@uwaterloo.ca}{\tt pierre.popoli@uwaterloo.ca}\\
\href{mailto:shallit@uwaterloo.ca}{\tt shallit@uwaterloo.ca}\\
\vphantom{space} 
\and
Department of Mathematics, University of York, Ian Wand
Building,
Deramore Lane,
York, YO10 5GH,
United Kingdom\\
\href{mailto:ingrid.vukusic@york.ac.uk}{\tt ingrid.vukusic@york.ac.uk}}

\maketitle

\begin{abstract}
The Fibonacci infinite word ${\bf f} = (f_i)_{i \geq 0} = 01001010\cdots$ is one of the most celebrated objects in combinatorics on words. There is a simple $5$-state automaton that, given $i$ in lsd-first Zeckendorf representation, computes its $i$'th term $f_i$, and a $2$-state automaton for msd-first. In this paper we consider the state complexity of the automaton generating the shifted sequence $(f_{i+c})_{i \geq 0}$, and show that it is $O(\log c)$ for both msd-first and lsd-first input.  This is close to the information-theoretic minimum for an aperiodic sequence. The techniques involve a mixture of state complexity techniques and Diophantine approximation. 
\end{abstract}

\section{Introduction}\label{sec:intro}

{\it State complexity} is one of the major themes of descriptional complexity.  We are given a regular language $L$, and we want to know what the size of the smallest finite automaton for $L$ is.  This topic has been studied for decades, starting with the fundamental papers of Maslov in the 1970's~\cite{Maslov:1970,Maslov:1973} and later popularized in the West by Yu, Zhuang, and Salomaa~\cite{Yu&Zhuang&Salomaa:1994}. 

Within state complexity, a major subtheme has explored bounds on the state complexity of various operations on regular languages, such as union, concatenation, and Kleene star.   Here we are given languages, say $L_1$ and $L_2$, accepted by automata of $m$ and $n$ states respectively, and want to know how large the minimal automata can be after applying some unary or binary operation to the languages \cite{Yu&Zhuang&Salomaa:1994}.

Cobham \cite{Cobham:1969,Cobham:1972} introduced the notion of {\it automatic sequence}: an automatic sequence $(a(i))_{i \geq 0}$ is a sequence over a finite alphabet $\Delta$ produced by a deterministic finite automaton with output (DFAO). We represent $i$ in some numeration system (e.g., base $k$ for $k \geq 2$) and feed the automaton with the digits $x$ of the representation. There is an output map $\tau$ from each state to $\Delta$, and $a(i)$ is then the value of $\tau$ applied to the last state reached on input $x$. There are two different variants: one where we read the representation starting with the most significant digit (called {\it msd-first}) and one starting with the least significant digit (called {\it lsd-first}). Although a sequence is automatic in the lsd-first sense if and only if it is automatic in the msd-first sense, the number of states required can be wildly different. More formally, a DFAO is a $6$-tuple $M=(Q,\Sigma,\delta,q_0,\Delta,\tau)$ where 
 $Q$ is a finite nonempty set of states;
$\Sigma$ is the input alphabet; 
$\delta:Q\times \Sigma \to Q$ is the transition function with domain extended to $Q\times \Sigma^*$ in the usual way;
$q_0$ is the initial state; 
$\Delta$ is the finite nonempty output alphabet; and
$\tau:Q\to \Delta$ is the output function. 
In this paper we assume that all DFAO's give the same correct output no matter how many leading zeros there are in the input (when reading the input in the msd-first sense), and the same thing must be true for
trailing zeros (in the lsd-first sense). 

We can, of course, combine the two topics, and study the state complexity of operations on automatic sequences, such as taking a linearly-indexed subsequence $(a(ci))_{i \geq 0}$, or shifting 
the sequence some number of terms
to get $(a(i+c))_{i \geq 0}$.  These topics were studied in \cite{Allouche&Shallit:2003}, and then more recently in a number of papers, such as
\cite{Zantema:2020,Charlier&Rampersad&Rigo&Waxweiler:2011,Zantema&Bosma:2022,Moradi&Rampersad&Shallit:2025}.

In this paper we focus on the operation of shifting a sequence  by some fixed number of terms $c$, to get a new sequence $(a(i+c))_{i \geq 0}$.  If $(a(i))_{i \geq 0}$ is generated by an $m$-state DFAO with lsd-first input in base $k$, we can construct a DFAO for its $c$-shift by using a deterministic transducer to compute the digit-by-digit sum $i+c$ on the fly, giving an automaton of $O(m \log c)$ states. However, this technique doesn't work for numbers in Zeckendorf representation (see Section~\ref{sec:zeck}), because there is no transducer that can do addition in one pass \cite{Labbe&Lepsova:2023}.  

In contrast with the lsd-first case in base $k$, the msd-first case can result in larger state complexity, because once again there is no deterministic transducer to compute the sum (consider adding $1$ to input $9999\cdots 9$ in base $10$).  
For example, for the celebrated Thue-Morse sequence ${\bf t} = (t(i))_{i \geq 0}$ \cite{Allouche&Shallit:1999}, it was recently shown that the msd-first state complexity of the DFAO for
$(t(i+c))_{i \geq 0}$ is $O(c)$, but exceeds $c^{0.694}$ infinitely often \cite{Moradi&Rampersad&Shallit:2025}.

However, in some cases, much lower complexity is seen for shifts of sequences in the msd-first case.  For example, if we consider the characteristic sequence ${\bf p}_2=(p_2(i))_i$ of the powers of $2$, defined by
$p_2 (i) = 1$ if $i$ is a power of $2$ and $0$ otherwise,
then it is possible to prove that the msd-first state complexity of the sequence $(p_2(i+c))_{i \geq 0}$ is $O(\log c)$ \cite{Moradi&Rampersad&Shallit:2025}.  This sort of logarithmic growth seems to be rather rare.  Indeed, logarithmic growth is close to the minimum possible blow-up for aperiodic shifted sequences, because a simple counting argument demonstrates that if an automatic sequence $(a(i))_{i \geq 0}$ is not eventually periodic, then its shifted sequence $(a(i+c))_{i \geq 0}$ must have state complexity exceeding $c' (\log c)/(\log \log c)$ infinitely often.\footnote{Specifically, the number of possible $n$-state DFAOs with input alphabet $\{0,1,\ldots k-1\}$ and output alphabet $\{0,1\}$ is at most
$n^{kn-1} 2^n$.  Either two different shifts are the same (which happens if and only if the sequence is eventually periodic), or there would have to be some shift  by $c \leq n^{kn-1} 2^n$ that could not be computed by a DFAO of $\leq n$ states.  Expressing $n$ in terms of $c$ now gives the bound.}
Thus the state complexity increase for shifts of ${\bf p}_2$ is quite close to the smallest conceivable blow-up.

In this paper we examine the state complexity of shifts of one particular famous aperiodic infinite word---the Fibonacci word ${\bf f} = f(0) f(1) f(2) \cdots = 01001010\cdots$. There are many equivalent ways to define $\bf f$, of which the most prominent seem to be the following: 
\begin{itemize} \setlength{\itemsep}{0em}
    \item[(a)]
It is the infinite fixed point, starting with $0$, of the morphism that sends $0$ to $01$ and
$1$ to $0$. 

\item[(b)] It is the infinite word arising from the limit of
$X_i$ as $i \rightarrow \infty$, where  $X_0 = 1$, $X_1 = 0$, and $X_i = X_{i-1} X_{i-2}$.

\item[(c)]  It is the Sturmian characteristic word of slope
$\gamma := 2-\phi$, where $\phi$ is the golden ratio $ (1+\sqrt{5})/2 \doteq 1.61803$. In other words, it is the word given by
$(\lfloor (i+2)\gamma \rfloor - \lfloor (i+1)\gamma \rfloor)_{i \geq 0}$.

\item[(d)] The $i$'th symbol of $\bf f$ is the last digit of the Zeckendorf representation of $i$. The Zeckendorf representation is recalled in~\cref{sec:zeck}.
\end{itemize}

There is yet a fifth way to define $\bf f$, which is not so well known, but will prove crucial to us in this paper.  It concerns fractional parts. Let $\{x\}$ denote the fractional part of $x$; that is, $x \bmod 1$.  Throughout the paper, we also use the  circular representation of intervals. If $\{x\} < \{y\}$, then the open interval $(x,y) \bmod 1$ is understood to mean $(\{ x \}, \{y\})$. And if $\{x\} > \{y\}$, then $(x,y) \bmod 1$ is understood to mean $ (\{ y \}, 1) \cup [0, \{x\} )$. Thus, for example, $(\pi,e) \bmod 1= (.14159\cdots, .71828\cdots)$ and $(e,\pi) \bmod 1 = (.71828\cdots, 1) \cup [0,.14159\cdots)$. Then we have the following fundamental lemma.

\begin{lemma} \label{lem:01-areas}
The $i$'th bit of ${\bf f}$, namely $f(i)$, is
$1$ if and only if 
$\{ i \phi  \} \in (-\phi, -2\phi) \bmod 1$.
\end{lemma} 

\begin{proof}
Let $\gamma = 2-\phi$.  Then from Item (c) of the equivalent definitions of ${\bf f}$ we have
\begin{align*}
f(i) = 1 \iff  \lfloor (i+2) \gamma \rfloor - \lfloor (i+1) \gamma \rfloor = 1 & \iff \{(i+1) \gamma \} \in (1-\gamma, 1), \\
& \iff  \{ (i+1) (2-\phi)\} \in (\phi - 1, 1), \\
& \iff \{ -(i+1) \phi \} \in (\phi-1 ,1), \\
& \iff  \{ (i+1)\phi \} \in (0, 1-\phi), \\
& \iff  \{ i \phi \} \in (-\phi, -2\phi) \bmod 1.
\end{align*}
\end{proof}

The Fibonacci word $\bf f$ is celebrated and well-studied \cite{Berstel:1986b}, so understanding the state complexity of its shifts is a natural and significant problem.  Our main result is that if we shift $\bf f$ by $c$, then the resulting word can be generated by a DFAO of only $O(\log c)$ states.  Furthermore, this bound holds if the input is read in either msd-first or lsd-first formats. We will need the following lemma, that follows directly from~\cref{lem:01-areas}. 

\begin{lemma} \label{lemma:f(i+c)-ialpha}
    We have $f(i+c) = 1 \iff \fract{i\phi} \in (-(c+1)\phi, -(c+2)\phi)\bmod 1$.
\end{lemma}
\begin{proof}
    By~\cref{lem:01-areas} we have 
    \begin{align*}
        f(i+c) = 1
    	& \iff
    	\fract{(i+c)\phi} \in (-\phi, -2\phi)\bmod 1, \\
        & \iff
    	\fract{i\phi} \in (-(c+1)\phi,-(c+2)\phi)\bmod 1. 
    \end{align*}
\end{proof}

While our main result concerns the state complexity of certain automata, the proof techniques are also of interest.  We use a mixture of arguments from Diophantine approximation, state complexity, and logic.  In particular, some parts of the proof are completed using the automated theorem-prover {\tt Walnut} \cite{Mousavi:2021,Shallit:2023}, which can give rigorous proofs of results about automatic sequences merely by stating the desired results in first-order logic.  In this conference submission version, many details are relegated to the Appendix.  


\section{Zeckendorf representations and Fibonacci-automatic sequences}\label{sec:zeck}

To state the results and proofs, we need a classic numeration system sometimes called the Zeckendorf numeration system~\cite{Lekkerkerker:1952,Zeckendorf:1972}, although Ostrowski~\cite{Ostrowski:1922} had already published a much more general numeration system in 1922. In this system, natural numbers are represented as sums of distinct Fibonacci numbers $F_i$ for $i \geq 2$, where as usual we write $F_0 = 0$, $F_1 = 1$, and $F_i = F_{i-1} + F_{i-2}$ for $i \geq 2$. To ensure unique representation, we demand that a Fibonacci number can be used at most once and no two consecutive Fibonacci numbers be used.  Thus, for example, $10$ is written as $8+2$ and not $5+5$ or $5+3+2$.

We write a representation for $n$ as a binary word over the alphabet $\{\tO,\tL\}$, denoted by $(n)_F$, of the form
$a_t a_{t-1} \cdots a_2$ where
$ n = \sum_{2 \leq i \leq t} a_i F_{i}$.
Thus, for example, the
representation of $17$ is $(17)_F=\tL\tO\tO\tL\tO\tL$. Notice that the notation $(\cdot)_F$ relies on msd-first representations. In a similar fashion, for a binary word $x$ of length $\ell$ of the form $x=x_{\ell-1}\cdots x_0$, we let $[x]_F$ denote the corresponding integer value; i.e., $[x]_F=\sum_{i=0}^{\ell-1} x_{i}F_{i+2}$. Then, for every integer $n\geq 0$, we have $[(n)_F]_F=n$, and for every binary word $x$ containing no $\tL \tL$, we have $([x]_F)_F=x$. 
The requirement that no two consecutive Fibonacci numbers are used then translates to the requirement that a valid Zeckendorf representation must not contain two consecutive $\tL$'s.

A sequence $(a(i))_{i \geq 0}$ is said to be {\it Fibonacci-automatic} if there exists a DFAO that takes as input a valid Zeckendorf representation of $i$ and outputs $a(i)$.   The Fibonacci word {\bf f} is the canonical example of a Fibonacci-automatic sequence; it is generated by the automata in~\cref{fib}.
\setlength{\intextsep}{-20pt plus 2pt minus 4pt}
\setlength{\belowcaptionskip}{30pt}
\begin{figure}[htb]
    \centering
\includegraphics[width=2.5in]{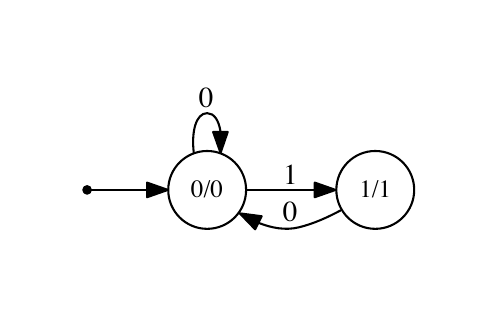}
\includegraphics[width=2.5in]{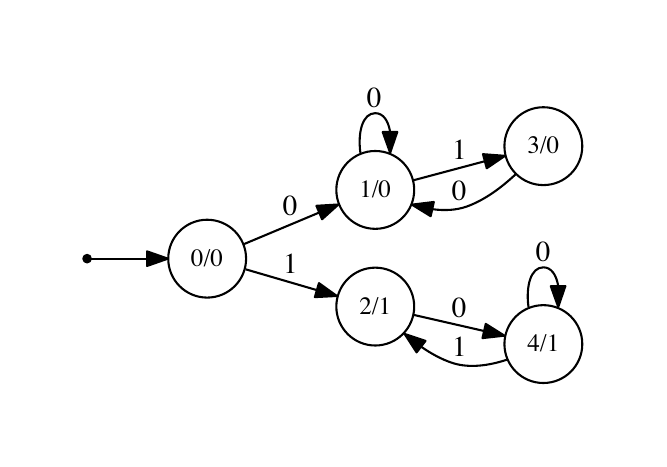}
    \vskip-.3in
    \caption{DFAOs for the Fibonacci word $\bf f$.  Msd-first on the left; lsd-first on the right.}
    \label{fib}
\end{figure}















Here, as usual, the label $a/b$ on a state means that $a$ is the state number and $b$ is the output associated with the state.
Notice that our convention is that inputs that form invalid Zeckendorf representations (i.e., those with two consecutive $1$'s) do not reach any state and do not produce any output.  

We will also need to represent negative numbers using Fibonacci numbers. Although many proposals exist for incorporating negative numbers into the Zeckendorf numeration system
\cite{Bunder:1992,Hajnal:2023,Shallit&Shan:2023,Labbe&Lepsova:2023}, we will need one that has not been discussed before.  It is, in fact, a special case of a more general representation by Rockett and Sz\"usz \cite{Rockett&Szusz:1992}.  Namely, to represent $-n$ for $n \geq 1$, we first find the smallest Fibonacci number $F_k \geq n$ and  write $n = F_k - m$.
We then express $m$ in the Zeckendorf system, say $\sum_{2 \leq i \leq t} a_i F_{i}$, and the representation of $-n$ is defined to be $(-F_k) + \sum_{2 \leq i \leq t} a_i F_{i}$. Notice that since $-F_{k}+F_{k-2}=-F_{k-1}$ for all $k\geq 3$, we have $t\leq k-3$. We can think of this representation, which we denote by $(-n)_F$, as a word over the alphabet $\{ -\tL, \tO, \tL\}$ where the leading digit is $-\tL$ and the remaining digits are either $\tO$ or $\tL$. The first few examples of these representations are given in~\cref{tab1}.
\setlength{\textfloatsep}{15pt plus 2pt minus 4pt}
\setlength{\intextsep}{15pt plus 2pt minus 4pt}
\setlength{\abovecaptionskip}{5pt}
\setlength{\belowcaptionskip}{15pt}
\begin{table}[H]
\centering
\begin{tabular}{c|c||c|c}
   $-n$ & $(-n)_F$ & $-n$ & $(-n)_F$ \\ \hline
   $-1$ & $(-\tL)$ & $-6$ & $(-\tL) \tO \tO \tL \tO$ \\
   $-2$ & $(-\tL) \tO$ &  $-7$ & $(-\tL) \tO \tO \tO \tL$ \\
   $-3$ & $(-\tL) \tO\tO$ & $-8$ & $(-\tL) \tO \tO \tO \tO$\\
   $-4$ & $(-\tL) \tO \tO \tL$ & $-9$ & $(-\tL) \tO \tO \tL \tO 1$ \\
   $-5$&  $(-\tL) \tO \tO \tO$ & $-1\tO$ & $(-1) \tO \tO \tL \tO \tO$ \\
\end{tabular}
\caption{Representation of the first few negative integers.}
\label{tab1}
\end{table}

\section{State complexity in the lsd-first case}\label{sec:lsd}

In this section we prove our first main result about the state complexity of the shifted sequence $(f(i+c))_i$ with lsd-first input. 
\begin{theorem} \label{theorem:sc-lsd}
    Let $(f(i))_{i \geq 0}$ be the Fibonacci word {\bf f}. For $c \geq 0$, there exists an lsd-first DFAO generating $(f(i+c))_{i \geq 0}$ with $O(\log c)$ states. 
\end{theorem}

The proof of~\cref{theorem:sc-lsd} is delayed to the end of this section. We first introduce a partition of the interval $[0,1)$ that is induced by the lsd-first representation.

\begin{definition} \label{def:I_d(m)}
For each $m$ with $0\leq m < F_{d+1}$ let $I_d(m)$ be the interval defined by $I_d(m)= [L_d(m) \phi, R_d(m) \phi) \bmod 1$ where if $d$ is even we have
\[
	R_d(m) = - F_{d+1} + m 
	\quad \text{and} \quad
	L_d(m) = \begin{cases}
		-F_d + m, &\text{if } m < F_d;\\
		- F_{d+2} + m, &\text{if } m \geq F_d,
	\end{cases}
\]
and if $d$ is odd, the values of $R_d(m)$ and $L_d(m)$ are switched.  
\end{definition}

Let $d\geq 2$. For every integer $n$,  let $n^{[\leq d]}$
denote the integer obtained from $n$ by discarding all digits with index larger than $d$; i.e., if $n = \sum_{2\leq i \leq t} a_i F_i$, then $n^{[\leq d]} = \sum_{i=2}^{\min\{d,t\}} a_i F_i$. Similarly, we define $n^{[\geq d]} = \sum_{i=\max\{2,d\}}^{t} a_i F_i$. Then clearly we have $n = n^{[\leq d]} + n^{[\geq d+1]}$.  The following lemma states that the intervals $I_d(m)$ for $0\leq m < F_{d+1}$ form a partition of $[0,1)$. Its proof and all necessary details are given in Appendix~\ref{appendix:lsd}. 

\begin{lemma}\label{lem:lsd-partition}
Let $d\geq 2$. Then the intervals $I_d(m)$ for $0\leq m < F_{d+1}$ form a partition of $[0,1)$.
Moreover, for every nonnegative integer $n$ we have
\begin{equation}\label{eq:partition_char}
    \fract{n\phi}\in I_d(m)
    \iff n^{[\leq d]} = m.
\end{equation}
\end{lemma}
Informally, the low-order Zeckendorf bits of an integer $n$ determine an interval where $n\phi$ lies, and the interval has length of order $\phi^{-d}$. 

\begin{remark}\label{rem:partition}
For every $d\geq 2$, each partition $\{I_{d+1}(m)\colon 0\leq m < F_{d+2}\}$ in~\cref{lem:lsd-partition} is a refinement of the partition $\{I_{d}(m)\colon 0\leq m < F_{d+1}\}$. This is clear from the construction of the intervals, as well as from \eqref{eq:partition_char}.
\end{remark}

\begin{remark}\label{rem:c_in_partition}
Let $c$ be a positive integer with $F_{k} < c \leq F_{k+1}$ for some $k\geq 2$ (or $c = 1$ and $k=2$). 
Set $m = F_{k+1} - c$ and note that $m < F_{k+1} - F_{k} = F_{k-1} < F_{k}$. Therefore,
we get that
\begin{align*}
    [L_{k}(m)\phi, R_{k}(m)\phi) \bmod 1 
    &= [ (- F_{k} + m)\phi, (-F_{k+1} + m)\phi )\bmod 1 \\
    &= [(-F_{k} + m)\phi, - c\phi ) \bmod 1,
\end{align*}
if $k$ is even. If $k$ is odd, we switch the endpoints of the interval.
In particular, the point $\fract{-c\phi}$ is an endpoint of some interval in the partition $\{I_{d}(m)\colon 0\leq m < F_{d+1}\}$ for all $d\geq k$. 
Moreover, one can check that by \cref{def:I_d(m)} the point $\fract{-c\phi}$ is not an endpoint in the partition for $d\leq k-1$.
\end{remark}

We are now ready to prove~\cref{theorem:sc-lsd}. 
\begin{proof}[of ~\cref{theorem:sc-lsd}]
We provide the construction of the automaton $M_c=(Q,  \Sigma, \delta, q_0, \Delta, \tau)$ for generating $f(i+c)$ with lsd-first input. Let $\Sigma=\{\tO,\tL\}$ and $\Delta = \{0,1\}$. 
We create the set of states $Q$ and the output function $\tau$ as follows. For $k\geq 0$, let $J_k=[-(k+1)\phi,-(k+2)\phi)\bmod 1$ and $\overline{J_k}=[-(k+2)\phi,-(k+1)\phi)\bmod 1$.
    
First, the initial state is denoted by $q_\varepsilon$.  Since the output associated with this state must match $f(c)$,~\cref{lem:01-areas} implies that $\tau(q_\varepsilon)=1$ if and only if $\fract{c\phi}\in J_0$.

Next, we add four states $A_{\tO},A_{\tL}$ and $R_{\tO},R_{\tL}$. Here $A_0$ and $A_1$ are sink states with output 1, and $R_0$ and $R_1$ are sink states with output 0; i.e., we have
\begin{align*}
    &\delta(A_{\tO},\tO)=A_{\tO}, \; \delta(A_{\tO},\tL)=A_{\tL}, \; \delta(A_{\tL},\tO)=A_{\tO}, \\ &\delta(R_{\tO},\tO)=R_{\tO}, \; \delta(R_{\tO},\tL)=R_{\tL}, \; \delta(R_{\tL},\tO)=R_{\tO}, \\
    &\tau(A_{\tO})=\tau(A_{\tL})=1, \; \tau(R_{\tO})=\tau(R_{\tL})=0. 
\end{align*} 
Notice that sink states are split in two due to the digit restrictions of the Zeckendorf numeration system; we are not permitted to simply self-loop on both $0$ and $1$ inputs.

The remaining states $q_x$ correspond to the state that the automaton reaches on input $x\in \{\tO,\tL\}^*$, given in lsd-first. 
Assume $F_{k}<c+2 \leq F_{k+1}$.  
Then from \cref{rem:c_in_partition} we know that both $\fract{-(c+1)\phi}$ and $\fract{-(c+2)\phi}$ are endpoints in the partition $\{I_k(m) \colon 0\leq m<F_{k+1}\}$.
Moreover, \cref{lem:lsd-partition} says that if we read the digits of an integer $n$ corresponding to $F_2, F_3, \ldots, F_k$, then we know which interval from the partition the point $\fract{n\phi}$ lies in. 
This means that we need to read the first $k-1$ digits.
After that we definitely know which interval from the partition $\{I_k(m) \colon 0\leq m<F_{k+1}\}$ the point $\fract{n\phi}$ lies in, and, in particular, whether $\fract{n\phi} \in J_c = [-(c+2)\phi,-(c+1)\phi)\bmod 1$ or not. 
Since by \cref{lemma:f(i+c)-ialpha} we have $f(n+c)=1 \iff \{n\phi\} \in J_c$, this means that the automaton reaches a sink state after reading at most $k-1$ digits. Namely, it reaches the state $A_{\tO}$ or $R_{\tO}$ (resp., $A_{\tL}$ or $R_{\tL}$) if the last input in lsd is $\tO$ (resp., $\tL$). 
Therefore, we only need to consider states which the automaton reaches on input $x$ with length $\leq k-1$.

To estimate the size of the resulting automaton, we count how many of these words $x$ of a given length $1\leq \ell \leq k-1$ cause the automaton to reach a state other than a sink state. 
First, we deal with the case $\ell=1$, where only two inputs are possible. We distinguish the following two cases:
\begin{itemize}

\item If $\fract{-(c+1)\phi}$ and $\fract{-(c+2)\phi}$ belong to $J_0$, then we add one state $q_\tO$. This case corresponds to $c$ such that $f(c-1)=f(c-2)=0$. The transitions from $q_\varepsilon$ are defined by $\delta(q_\varepsilon,\tO)=q_\tO$ and $\delta(q_\varepsilon,\tL)=A_\tL$ if $f(c+1)=1$ or $\delta(q_\varepsilon,\tL)=R_\tL$ if $f(c+1)=0$. 

\item Otherwise, we add two states $q_\tO$ and $q_\tL$, and transitions from $q_\varepsilon$ are defined by $\delta(q_\varepsilon,a)=q_a$ for $a\in \{\tO,\tL\}$. 
\end{itemize}
Now suppose that $2 \leq \ell \leq k-1$. 
For an input word $x\in \{\tO,\tL\}^\ell$,
let $m=[x]_F$ and $I_{\ell + 1}(m)$ be the interval defined in~\cref{def:I_d(m)}. 
We distinguish two cases, as follows: 
\begin{itemize}
\item If $I_{\ell + 1}(m) \subseteq J_c$, then by~\cref{lemma:f(i+c)-ialpha} the automaton $M_c$ reaches the sink state $A_i$ after reading $x$, where $i$ is the last letter of $x$. Similarly, if $I_{\ell + 1}(m) \subseteq \overline{J_c}$, then  $M_c$ reaches the sink state $R_i$ after reading $x$, where $i$ is the last letter of $x$.

\item Otherwise, exactly one of the two points $\{-(c+2)\phi\}$ or $\{-(c+1)\phi\}$ is included in $I_{\ell + 1}(m)$. In this case, we add a new state $q_x$ to $Q$, and the output of $q_x$ is given by $\tau(q_x) = 1$ if and only if $\{m\phi\} \in J_c$. Since the set $\{I_{\ell + 1}(m)\colon 0\leq m < F_{\ell + 2}\}$ is a partition by~\cref{lem:lsd-partition}, this case can occur at most twice for each length $\ell$. 
\end{itemize} 

Therefore, for each given length $1\leq \ell \leq k-1$, at most two non-sink states are added. The transition function on these states is defined by $\delta(q_x,a)=q_{xa}$.  So the total number of states in $Q$ is at most $1 + 4 + 2(k-1) = 2k + 3$, which is indeed $O(\log c)$, since $F_k < c+2 \leq F_{k+1}$. 


\end{proof}

\begin{example}
We illustrate the proof of the previous theorem with $c=10$. 
In the bottom line of the upper part of \cref{fig:lsd-c=10} one can see that $\fract{-12 \phi}$ and $\fract{-11\phi}$ are endpoints of the partition obtained after reading $5$ digits of input. 
Since $c+2 = 12\leq 13 = F_7$, Remark~\ref{rem:c_in_partition} guarantees that this happens indeed after reading $5$ digits.
The lower part of \cref{fig:lsd-c=10} illustrates the fact that we want to output $1$ if and only if $\fract{n\alpha}\in [-11\phi, -12\phi) \bmod 1$, which is the blue area in the top line. Intervals that are fully contained in $[-11\phi, -12\phi) \bmod 1$ are colored blue. The blue intervals signify that after having read the corresponding digits we can go to one of the sink states $A_{\tO},A_{\tL}$. Similarly, intervals that are fully contained in $[-12\phi, -11\phi) \bmod 1$ are colored orange, and are associated with the sink states $R_{\tO},R_{\tL}$. Intervals that contain one of the ``critical points'' $\fract{-11\phi}, \fract{-12\phi}$ are colored black. These intervals correspond to the ``intermediate states'' listed below.

\setlength{\belowcaptionskip}{0pt}
\begin{figure*}[ht]
\centering
\scalebox{1}{
\begin{tikzpicture}

\draw[thick] (0,0) -- (12,0);
\node at (0,0.5){\tiny{$\{-\phi\}$}};
\draw (0,0.15)-- (0,-0.15);
\draw[dashed] (0,0.15) -- (0,-5.15);
\node at (12,0.5){\tiny{$\{-\phi\}$}};
\draw (12,0.15)-- (12,-0.15);
\draw[dashed] (12,0.15) -- (12,-5.15);
\node at (4.58,0.5){\tiny{$\{-2\phi\}$}};
\draw (4.58,-0.15)-- (4.58,0.15);
\draw[dashed] (4.58,0.15) -- (4.58,-5.15);
\draw[<->] (0,0.15) to (4.58,0.15);
\node at (2.29,0.35){\footnotesize{$\tL$}};
\draw[<->] (12,0.15) to (4.58,0.15);
\node at (8.19,0.35){\footnotesize{$\tO$}};
\draw[thick] (0,-1.25) -- (12,-1.25);
\node at (9.16,-0.65){\tiny{$\{-3\phi\}$}};
\draw (9.16,-1.40)-- (9.16,-1.10);
\draw[dashed] (9.16,-1.40) -- (9.16,-5.15);
\draw[<->] (0,-1.10) to (4.58,-1.10);
\node at (2.29,-0.9){\footnotesize{$\tL\tO$}};
\draw[<->] (4.58,-1.10) to (9.16,-1.10);
\node at (6.87,-0.9){\footnotesize{$\tO\tO$}};
\draw[<->] (12,-1.10) to (9.16,-1.10);
\node at (10.58,-0.9){\footnotesize{$\tO\tL$}};
\draw[thick] (0,-2.5) -- (12,-2.5);
\node at (1.75,-2){\tiny{$\{-4\phi\}$}};
\draw (1.75,-2.65)-- (1.75,-2.35);
\draw[dashed] (1.75,-2.65) -- (1.75,-5.15);
\draw[<->] (0,-2.35) to (1.75,-2.35);
\node at (0.87,-2.15){\footnotesize{$\tL\tO\tL$}};
\draw[<->] (1.75,-2.35) to (4.58,-2.35);
\node at (3.17,-2.15){\footnotesize{$\tL\tO\tO$}};
\node at (6.33,-2){\tiny{$\{-5\phi\}$}};
\draw (6.33,-2.65)-- (6.33,-2.35);
\draw[dashed] (6.33,-2.65) -- (6.33,-5.15);
\draw[<->] (4.58,-2.35) to (6.33,-2.35);
\node at (5.45,-2.15){\footnotesize{$\tO\tO\tL$}};
\draw[<->] (6.33,-2.35) to (9.16,-2.35);
\node at (7.74,-2.15){\footnotesize{$\tO\tO\tO$}};
\draw[<->] (9.16,-2.35) to (12,-2.35);
\node at (10.58,-2.15){\footnotesize{$\tO\tL\tO$}};
\draw[thick] (0,-3.75) -- (12,-3.75);
\draw[<->] (0,-3.60) to (1.75,-3.60);
\node at (0.87,-3.40){\footnotesize{$\tL\tO\tL\tO$}};
\draw[<->] (4.58,-3.60) to (6.33,-3.60);
\node at (5.45,-3.40){\footnotesize{$\tO\tO\tL\tO$}};
\node at (10.91,-3){\tiny{$\{-6\phi\}$}};
\draw[dashed] (10.91,-3.60) -- (10.91,-5.15);
\draw[<->] (9.16,-3.60) to (10.91,-3.60);
\node at (10,-3.40){\footnotesize{$\tO\tL\tO\tO$}};
\draw[<->] (10.91,-3.60) to (12,-3.60);
\node at (11.5,-3.40){\footnotesize{$\tO\tL\tO\tL$}};
\node at (3.50,-3){\tiny{$\{-7\phi\}$}};
\draw[dashed] (3.50,-3.60) -- (3.50,-5.15);
\draw[<->] (1.75,-3.60) to (3.50,-3.60);
\node at (2.65,-3.40){\footnotesize{$\tL\tO\tO\tO$}};
\draw[<->] (3.50,-3.60) to (4.58,-3.60);
\node at (4.04,-3.40){\footnotesize{$\tL\tO\tO\tL$}};
\node at (8.08,-3){\tiny{$\{-8\phi\}$}};
\draw[dashed] (8.08,-3.60) -- (8.08,-5.15);
\draw[<->] (6.33,-3.60) to (8.08,-3.60);
\node at (7.20,-3.40){\footnotesize{$\tO\tO\tO\tO$}};
\draw[<->] (8.08,-3.60) to (9.16,-3.60);
\node at (8.66,-3.40){\footnotesize{$\tO\tO\tO\tL$}};
\draw[thick] (0,-5) -- (12,-5);
\draw[<->] (3.50,-4.85) to (4.58,-4.85);
\node at (4.05,-5.35){\footnotesize{$\tL\tO\tO\tL\tO$}};
\draw[<->] (8.08,-4.85) to (9.16,-4.85);
\node at (8.66,-5.65){\footnotesize{$\tO\tO\tO\tL\tO$}};
\draw[<->] (10.91,-4.85) to (12,-4.85);
\node at (11.5,-5.35){\footnotesize{$\tO\tL\tO\tL\tO$}};
\node at (0.66,-4.5){\tiny{$\{-9\phi\}$}};
\draw (0.66,-5.15)-- (0.66,-4.85);
\draw[<->] (0,-4.85) to (0.66,-4.85);
\draw[<->] (0.66,-4.85) to (1.75,-4.85);
\node at (0.33,-5.35){\footnotesize{$\tL\tO\tL\tO\tL$}};
\node at (1.3,-5.65){\footnotesize{$\tL\tO\tL\tO\tO$}};
\node at (5.25,-4.5){\tiny{$\{-10\phi\}$}};
\draw (5.25,-5.15)-- (5.25,-4.85);
\draw[<->] (4.58,-4.85) to (5.25,-4.85);
\draw[<->] (5.25,-4.85) to (6.33,-4.85);
\node at (4.9,-5.65){\footnotesize{$\tO\tO\tL\tO\tL$}};
\node at (5.8,-5.35){\footnotesize{$\tO\tO\tL\tO\tO$}};
\node at (9.83,-4.5){\tiny{$\{-11\phi\}$}};
\draw (9.83,-5.15)-- (9.83,-4.85);
\draw[<->] (9.16,-4.85) to (9.83,-4.85);
\draw[<->] (9.83,-4.85) to (10.91,-4.85);
\node at (9.65,-5.35){\footnotesize{$\tO\tL\tO\tO\tL$}};
\node at (10.55,-5.65){\footnotesize{$\tO\tL\tO\tO\tO$}};
\node at (2.41,-4.5){\tiny{$\{-12\phi\}$}};
\draw (2.41,-5.15)-- (2.41,-4.85);
\draw[<->] (1.75,-4.85) to (2.41,-4.85);
\draw[<->] (2.41,-4.85) to (3.50,-4.85);
\node at (3.1,-5.65){\footnotesize{$\tL\tO\tO\tO\tO$}};
\node at (2.2,-5.35){\footnotesize{$\tL\tO\tO\tO\tL$}};
\node at (7.00,-4.5){\tiny{$\{-13\phi\}$}};
\draw (7.00,-5.15)-- (7.00,-4.85);
\draw[<->] (6.33,-4.85) to (7.00,-4.85);
\draw[<->] (7.00,-4.85) to (8.08,-4.85);
\node at (6.7,-5.65){\footnotesize{$\tO\tO\tO\tO\tL$}};
\node at (7.65,-5.35){\footnotesize{$\tO\tO\tO\tO\tO$}};


\node at (6.7,-4.65-8){\textcolor{white}{\footnotesize{$\tO\tO\tO\tO\tL$}}};
\draw[thick] (-0.15,1-8) -- (12.15,1-8);
\node at (0,1.5-8){\tiny{$\{-\phi\}$}};
\draw (0,1.15-8)-- (0,0.85-8);
\node at (12,1.5-8){\tiny{$\{-\phi\}$}};
\draw (12,1.15-8)-- (12,0.85-8);
\node(a11) at (9.83,1.5-8){\tiny{$\{-11\phi\}$}};
\draw (9.83,1.15-8)-- (9.83,0.85-8);
\node(a12) at (2.41,1.5-8){\tiny{$\{-12\phi\}$}};
\draw (2.41,1.15-8)-- (2.41,0.85-8);
\draw[dashed] (0,1.15-8) -- (0,-4.15-8);
\draw[dashed] (12,1.15-8) -- (12,-4.15-8);
\draw[<->,orange,thick] (2.41,1.15-8) to (9.83,1.15-8);
\node[orange] at (6.12,1.35-8){\tiny{Output $0$}};
\draw[<->,cyan,thick] (0,1.15-8) to (2.41,1.15-8);
\node[cyan] at (1.2,1.35-8){\tiny{Output $1$}};
\draw[<->,cyan,thick] (9.83,1.15-8) to (12,1.15-8);
\node[cyan] at (10.92,1.35-8){\tiny{Output $1$}};
\draw[thick] (-0.15,0-8) -- (12.15,0-8);
\draw (4.58,-0.15-8)-- (4.58,0.15-8);
\draw[dashed] (4.58,0.15-8) -- (4.58,-4.15-8);
\draw[<->] (0,0.15-8) to (4.58,0.15-8);
\draw[<->] (12,0.15-8) to (4.58,0.15-8);
\draw[thick] (-0.15,-1-8) -- (12.15,-1-8);
\draw (9.16,-1.15-8)-- (9.16,-0.85-8);
\draw[dashed] (9.16,-1.15-8) -- (9.16,-4.15-8);
\draw[<->] (0,-0.85-8) to (4.58,-0.85-8);
\draw[<->,orange] (4.58,-0.85-8) to (9.16,-0.85-8);
\draw[<->] (12,-0.85-8) to (9.16,-0.85-8);
\draw[thick] (-0.15,-2-8) -- (12.15,-2-8);
\draw (1.75,-2.15-8)-- (1.75,-1.85-8);
\draw[dashed] (1.75,-2.15-8) -- (1.75,-4.15-8);
\draw[<->,cyan] (0,-1.85-8) to (1.75,-1.85-8);
\draw[<->] (1.75,-1.85-8) to (4.58,-1.85-8);
\draw[<->,orange] (4.58,-1.85-8) to (9.16,-1.85-8);
\draw[<->] (9.16,-1.85-8) to (12,-1.85-8);
\draw[thick] (-0.15,-3-8) -- (12.15,-3-8);
\draw[<->,cyan] (0,-2.85-8) to (1.75,-2.85-8);
\draw (10.91,-3.15-8)-- (10.91,-2.85-8);
\draw[dashed] (10.91,-3.15-8) -- (10.91,-4.15-8);
\draw[<->] (9.16,-2.85-8) to (10.91,-2.85-8);
\draw[<->,cyan] (10.91,-2.85-8) to (12,-2.85-8);
\draw (3.50,-3.15-8)-- (3.50,-2.85-8);
\draw[dashed] (3.50,-3.15-8) -- (3.50,-4.15-8);
\draw[<->] (1.75,-2.85-8) to (3.50,-2.85-8);
\draw[<->,orange] (3.50,-2.85-8) to (4.58,-2.85-8);
\draw[<->,orange] (4.58,-2.85-8) to (9.16,-2.85-8);
\draw[thick] (-0.15,-4-8) -- (12.15,-4-8);
\draw[<->,orange] (3.50,-3.85-8) to (4.58,-3.85-8);
\draw[<->,cyan] (10.91,-3.85-8) to (12,-3.85-8);
\draw[<->,cyan] (0,-3.85-8) to (1.75,-3.85-8);
\draw (9.83,-4.15-8)-- (9.83,-3.85-8);
\draw[<->,orange] (9.16,-3.85-8) to (9.83,-3.85-8);
\draw[<->,cyan] (9.83,-3.85-8) to (10.91,-3.85-8);
\draw (2.41,-4.15-8)-- (2.41,-3.85-8);
\draw[<->,cyan] (1.75,-3.85-8) to (2.41,-3.85-8);
\draw[<->,orange] (2.41,-3.85-8) to (3.50,-3.85-8);
\draw[<->,orange] (4.58,-3.85-8) to (9.16,-3.85-8);

\tikzset{cross/.style={cross out, draw=cyan, fill=none, minimum size=2*(#1-\pgflinewidth), inner sep=0pt, outer sep=0pt}, cross/.default={2pt}}

\draw node[cross] at (10.25,-0.85-8){};
\node[cyan] at (10.25,-0.65-8){\tiny{$\{2\phi\}$}};
\draw node[cross] at (10.25,-1.85-8){};
\node[cyan] at (10.25,-1.65-8){\tiny{$\{2\phi\}$}};
\draw node[cross] at (10.25,-2.85-8){};
\node[cyan] at (10.25,-2.65-8){\tiny{$\{2\phi\}$}};

\tikzset{cross/.style={cross out, draw=orange, fill=none, minimum size=2*(#1-\pgflinewidth), inner sep=0pt, outer sep=0pt}, cross/.default={2pt}}

\draw node[cross] at (2.83,0.15-8){};
\node[orange] at (2.83,0.35-8){\tiny{$\{\phi\}$}};
\draw node[cross] at (7.41,0.15-8){};
\node[orange] at (7.41,0.35-8){\tiny{$\{0\}$}};
\draw node[cross] at (2.83,-0.85-8){};
\node[orange] at (2.83,-0.65-8){\tiny{$\{\phi\}$}};
\draw node[cross] at (2.83,-1.85-8){};
\node[orange] at (2.83,-1.65-8){\tiny{$\{\phi\}$}};
\draw node[cross] at (2.83,-2.85-8){};
\node[orange] at (2.83,-2.65-8){\tiny{$\{\phi\}$}};

\end{tikzpicture}
}


\caption{Illustrations of the partition lemma and construction of the automaton for $c=10$.}
\label{fig:lsd-c=10}
\end{figure*}
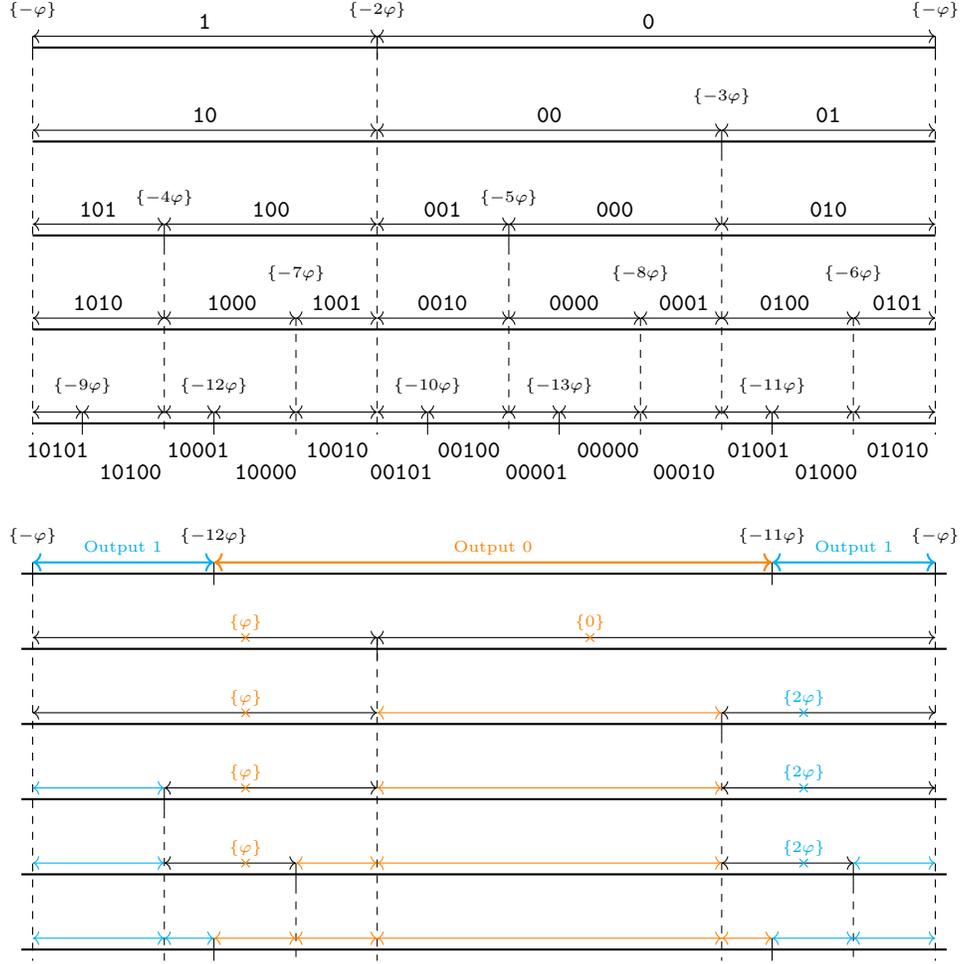

Overall, the states are the following:
\begin{itemize} \setlength{\itemsep}{0em}
    \item Initial state $q_\varepsilon$, with output $0$, since $\fract{0\phi} \notin [-11\phi, -12 \phi)$. 
    \item Intermediate states:~$q_\tO,q_\tL,q_{\tO\tL},q_{\tL\tO},q_{\tL\tO\tO},q_{\tO\tL\tO},q_{\tL\tO\tO\tO},q_{\tO\tL\tO\tO}$. The output of each intermediate state can be computed explicitly, and can be seen in the lower part figure of \cref{fig:lsd-c=10} (via the orange/blue points $\fract{\phi}, \fract{0}, \fract{2\phi}$). Indeed, it suffices to compute the value of $\fract{[x^R]_F\phi}$ for each $q_x$, where $z^R$ is the reverse image of $z$, and check if it belongs to the interval $[-11\phi,-12\phi)\bmod 1$ to output $1$.
    \item Two pairs of sink states: $A_{\tO},A_{\tL}$ and $R_{\tO},R_{\tL}$.
\end{itemize}
The automaton is presented in~\cref{fig:lsd-automaton-c=10}.


\begin{figure}[ht]
\centering
\begin{tikzpicture}
\tikzstyle{every node}=[shape=circle,fill=none,draw=black,minimum size=35pt,inner sep=2pt]

\node(a) at (0,0) {$q_\varepsilon/0$};

\node(0) at (1,1.75) {\footnotesize{$q_\tO$}$/0$};
\node(1) at (1,-1.75) {\footnotesize{$q_\tL$}$/0$};

\node(01) at (3,1.75) {\footnotesize{$q_{\tO\tL}$}$/1$};
\node(10) at (3,-1.75) {\footnotesize{$q_{\tL\tO}$}/0};

\node(010) at (5,1.75) {\footnotesize{$q_{\tO\tL\tO}$}/1};
\node(100) at (5,-1.75) {\footnotesize{$q_{\tL\tO\tO}$}/0};

\node(0100) at (7,1.75) {\footnotesize{$q_{\tO\tL\tO\tO}$}/1};
\node(1000) at (7,-1.75) {\footnotesize{$q_{\tL\tO\tO\tO}$}/0};

\node(A0) at (6,0) {$A_{\tO}/1$};
\node(A1) at (4,0) {$A_{\tL}/1$};
\node(R0) at (9.5,1) {$R_{\tO}/0$};
\node(R1) at (9.5,-1) {$R_{\tL}/0$};

\tikzstyle{every node}=[shape=circle,fill=none,draw=none,minimum size=10pt,inner sep=2pt]
\node(a0) at (-1.5,0) {};

\tikzstyle{every path}=[color =black, line width = 0.5 pt]
\tikzstyle{every node}=[shape=circle,minimum size=5pt,inner sep=2pt]
\draw [->] (a0) to [] node [] {}  (a);

\draw [->] (a) to [] node [left] {$\tO$}  (0);
\draw [->] (a) to [] node [left] {$\tL$}  (1);

\draw [->] (0) to [bend left=33] node [above] {$\tO$}  (R0);
\draw [->] (0) to [] node [above] {$\tL$}  (01);

\draw [->] (1) to [] node [above] {$\tO$}  (10);

\draw [->] (01) to [] node [above] {$\tO$}  (010);

\draw [->] (10) to [] node [above] {$\tO$}  (100);
\draw [->] (10) to [] node [left] {$\tL$}  (A1);

\draw [->] (010) to [] node [above] {$\tO$}  (0100);
\draw [->] (010) to [] node [above] {$\tL$}  (A1);

\draw [->] (100) to [] node [above] {$\tO$}  (1000);
\draw [->] (100) to [bend right=38] node [below] {$\tL$}  (R1);

\draw [->] (1000) to [] node [below] {$\tO$}  (R0);
\draw [->] (1000) to [bend left=5] node [above] {$\tL$}  (A1);

\draw [->] (0100) to [] node [above left] {$\tO$}  (A0);
\draw [->] (0100) to [] node [above ] {$\tL$}  (R1);

\draw [->] (A0) to [loop right] node [right] {$\tO$}  ();

\draw [->] (A0) to [bend left=10] node [below] {$\tL$}  (A1);
\draw [->] (A1) to [bend left=10] node [above] {$\tO$}  (A0);

\draw [->] (R0) to [loop right] node [right] {$\tO$}  ();

\draw [->] (R0) to [bend left=10] node [right] {$\tL$}  (R1);
\draw [->] (R1) to [bend left=10] node [left] {$\tO$}  (R0);

;
\end{tikzpicture}
\caption{Lsd-first DFAO for $(f(i+10))_{i \geq 0}$.}
\label{fig:lsd-automaton-c=10}
\end{figure}
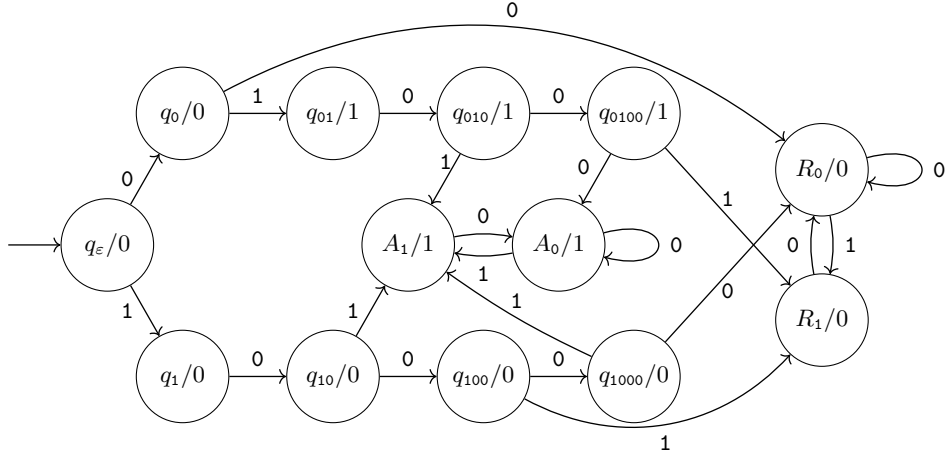


\end{example}

The next theorem provides an exact formula for the state complexity of the sequence $(f(i+c))_{i \geq 0}$. 
In particular, it implies that the upper bound in \cref{theorem:sc-lsd} is optimal.
The proof of~\cref{thm:lsd-first-optimality} is given in Appendix~\cref{appendix:lsd}.

\begin{theorem} \label{thm:lsd-first-optimality}
    Let $(f(i))_{i \geq 0}$ be the Fibonacci word {\bf f}. For $c \geq 5$, the number of states in the minimal lsd-first DFAO generating $(f(i+c))_{i \geq 0}$ is $2(|(c)_F|+1)+1-g(c)$, where $g(c)=1$ if $f(c-1)=f(c-2)=0$, and $g(c)=0$ otherwise. In particular, the function that maps $c$ to the state complexity of $(f(i+c))_{i \geq 0}$ is Fibonacci-regular. 
\end{theorem}

\section{State complexity in the msd-first case}\label{sec:msd}

We now turn to the same problem for the msd-first format. To determine a DFAO that generates the shifted Fibonacci word $(f(i+c))_i$ in msd-first format, one can make use of the construction of~\cref{theorem:sc-lsd} based on lsd-first format. To do so, we reverse all transitions in the automaton and determinize the resulting new automaton. This operation could potentially entail an exponential blowup of the number of states, and therefore we would only get a bound of $O(c^e)$ states for msd-first format, for some $e>0$. However, in this section, we prove the following much stronger theorem.
\begin{theorem}\label{thm:msd_automaton}
For all integers $c \geq 0$ there is a Fibonacci DFAO of $O(\log c)$ states that generates the shifted Fibonacci word $(f(i+c))_{i \geq 0}$ in msd-first format.
\end{theorem}

The proof of~\cref{thm:msd_automaton} is delayed to the end of this section. As before, we first introduce a partition that is relevant for the msd-first representation. Let $n = \sum_{2 \leq i \leq t} a_i F_{i}$ be the Fibonacci representation of $n$. For $a \in \{0,1\}$ we write 
\begin{equation}\label{eq:append_delta_pos}
    n^{(a)} = \sum_{i=3}^{t+1} a_{i-1} F_{i} + a F_2.
\end{equation} 
In other words, the operation $n^{(a)}$ appends the bit $a$ at the right end of the Zeckendorf expansion of $n$. Note that if $a_2=1$, then $n^{(1)}$ is not a valid Fibonacci representation. For a negative number $m=-F_k + \sum_{2 \leq i \leq t} a_i F_{i}$, we similarly define
\begin{equation}\label{eq:append_delta_neg}
    m^{(a)}
    = - F_{k+1} + \sum_{i=3}^{t+1} a_{i-1} F_i + a F_2, 
\end{equation} and again the above representation might not be a valid representation if $a = 1$. 

\begin{table}[H]
\centering
\begin{tabular}{c|l|l|l}
$n$ & $(n)_F$ & $(n^{(0)})_F$ & $(n^{(1)})_F$ \\ \hline
$1$ & $\tL$ & $\tL\tO$ & \, - \\ 
$2$ & $\tL\tO$ & $\tL\tO\tO$ & $\tL\tO\tL$ \\ 
$3$ & $\tL\tO\tO$ & $\tL\tO\tO\tO$ & $\tL\tO\tO\tL$ \\
$4$ & $\tL\tO\tL$ & $\tL\tO\tL\tO$ & \, - \\
$5$ & $\tL\tO\tO\tO$ & $\tL\tO\tO\tO\tO$ & $\tL\tO\tO\tO\tL$ \\
$6$ & $\tL\tO\tO\tL$ & $\tL\tO\tO\tL\tO$ & \, -  \\
\end{tabular}
\caption{Examples of the $(\cdot)^{(a)}$ operation.}
\end{table}

We will also need to shift these representations (both positive and negative) one place to the right, discarding the rightmost bit in msd-first representation. We denote this operation with the symbol $'$, as in $n'$. Therefore, if $n=\sum_{2 \leq i \leq t} a_i F_{i}$, then $n'=\sum_{2 \leq i \leq t-1} a_{i+1} F_{i}$. Similarly, if $-n=(-F_k)+\sum_{2 \leq i \leq t} a_i F_{i}$, then $(-n)'=(-F_{k-1})+\sum_{2 \leq i \leq t-1} a_{i+1} F_{i}$. Note that for $m = -1 = -F_2$, this definition gives $(-1)' = -F_1=-1$. A brief table of this operation is in~\cref{tab2}.

\begin{table}[H]
\centering
\begin{tabular}{c|l||c|l||c|l}
$n$ & $(n)_F$ & $n'$ & $(n')_F$ & $(-n)'$ & $((-n)')_F$ \\
\hline
$1$ & $\tL$ & $0$ & $\tO$ & $-1$& $(-\tL)$  \\
$2$ & $\tL\tO$ & $1$ & $\tL$ & $-1$ & $(-\tL)$  \\
$3$ & $\tL\tO\tO$ & $2$ & $\tL\tO$ & $-2$ & $(-\tL) \tO$\\
$4$ & $\tL\tO\tL$ & $2$ & $\tL\tO$ & $-3$ & $(-\tL) \tO\tO$\\
$5$ & $\tL\tO\tO\tO$ & $3$ & $\tL\tO\tO$ & $-3$ & $(-\tL) \tO\tO$\\
$6$ & $\tL\tO\tO\tL$ & $3$ & $\tL\tO\tO$ & $-4$ & $(-\tL) \tO\tO\tL$\\
$7$ & $\tL\tO\tL\tO$ & $4$ & $\tL\tO\tL$ & $-5$ & $(-\tL) \tO\tO\tO$\\
$8$ & $\tL\tO\tO\tO\tO$ & $5$ & $\tL\tO\tO\tO$ & $-5$ & $(-\tL) \tO\tO\tO$\\
$9$ & $\tL\tO\tO\tO\tL$ & $5$ & $\tL\tO\tO\tO$ & $-6$ & $(-\tL) \tO\tO\tL\tO$\\
\end{tabular}
\caption{Examples of the ${}'$ operation.}
\label{tab2}
\end{table}

\begin{definition}
Let $\cP$ be a partition of the set $[0,1)$ into intervals modulo $1$.
We say that $\cP$ is \emph{consistent with the msd-first Fibonacci representation}
if for every interval $I\in \cP$ there exist two intervals $I_0, I_1 \in \cP$ such that for all nonnegative integers $n$ we have
\[
    \fract{n \phi} \in I 
    \implies 
    \begin{cases}
        \fract{\nzero \phi} \in I_0, &\text{if } f(n) = 1;\\
        \fract{\nzero \phi} \in I_0 \text{ and } \fract{\none \phi} \in I_1,
        &\text{otherwise}.
    \end{cases}
\]
\end{definition}

Note that in the above definition the condition $f(n) = 1$ is equivalent to the least significant digit in the Zeckendorf representation of $n$ being $1$; see Definition (d) of the Fibonacci word in Section~\ref{sec:intro}.

\begin{lemma} \label{lemma:partition-intersection-property}
    Let $\cP_A$ and $\cP_B$ be two partitions of $[0,1)$ that are consistent with the msd-first Fibonacci representation. Then the partition $\cP_{A,B}$ obtained from combining the endpoints of both $\cP_A$ and $\cP_B$ is also consistent with the msd-first Fibonacci representation.
\end{lemma}

\begin{proof}
    Consider an interval $I \in \cP_{A,B}$. Then $I$ can be written as the intersection of some intervals $I_A \in \cP_A$ and $I_B \in \cP_B$, even if both endpoints of $I$ are from a single partition; in other words, $I = I_A \cap I_B$. Therefore, by assumption there are intervals $I_{A, a} \in \cP_A, I_{B, a} \in \cP_B$ for $a \in \{0, 1\}$ such that for all $n$ where $\fract{n\phi} \in I$, we have $\{\ndelta \phi\} \in I_{A, a} \cap I_{B, a} \in \cP_{A,B}$. Thus, the new partition is consistent with the msd-first Fibonacci representation.
\end{proof}

\begin{definition} \label{def:msd-partition}
    Let $r\leq -2$ be a negative integer with representation $r = - F_R + \sum_{i=2}^{R-3} b_i F_i$. We define the partition $\cP(r)$ of $[0,1)$ as follows. First, set $r_0 = r$ and recursively define $r_{k} = r_{k-1}'$ for $1 \leq k \leq R-2$, or equivalently
\[
    r_k = -F_{R-k} + \sum_{i=2}^{R-3-k} b_{i+k} F_{i},
    \quad 0\leq k \leq R-2.
\] Then we define $\cP(r)$ to be the partition of the set $[0,1)$ into intervals modulo $1$, where the endpoints are exactly the points $\fract{r_k \phi}$ for $0 \leq k \leq R-2$. Note that we always have $r_{R-2} = -1$, $r_{R-3} = -2$, and $r_{R-4} = -3$ if $r \leq -3$.

\end{definition}

\begin{example}
   We illustrate the partition $\cP(r)$ with the example $r = -14 = -F_8 + F_5 + F_3$:
\begin{align*}
    r_0 &= -F_8 + F_5 + F_3 = -14,  && \fract{-14\phi}=0.347 \cdots, \\
    r_1 &= -F_7 + F_4 + F_2 = - 9, && \fract{-9\phi}=0.437 \cdots,\\
    r_2 &= -F_6 + F_3 = - 6, && \fract{-6\phi}=0.291 \cdots, \\
    r_3 &= -F_5 + F_2 = - 4, && \fract{-4\phi}=0.527 \cdots, \\
    r_4 &= -F_4 = - 3, && \fract{-3\phi}=0.145 \cdots, \\
    r_5 &= -F_3 = - 2, && \fract{-2\phi}=0.763 \cdots, \\
    r_6 &= -F_2 = - 1, && \fract{-\phi}=0.381 \cdots.
\end{align*} 
Then the partition $\cP(-14)$ is composed of the  following $7$ intervals:
\begin{align*}
&[-3\phi,-6\phi) \bmod 1,\ [-6\phi,-14\phi)\bmod 1,\ [-14\phi,-\phi)\bmod 1, 
[-\phi,-9\phi)\bmod 1, \\ 
&[-9\phi,-4\phi)\bmod 1,\
[-4\phi,-2\phi)\bmod 1,\ [-2\phi,-3\phi)\bmod 1.
\end{align*}
\end{example}

Now we state our main lemma of this section. Its proof is given in Appendix~\ref{appendix:msd}. 
\begin{lemma}\label{lem:ingrid}
    Let $r\leq -2$ be a negative integer and $\cP(r)$ the partition from \cref{def:msd-partition}.
    Then $\cP(r)$ is consistent with the msd-first Fibonacci representation. 
\end{lemma}

We can now turn to the proof of~\cref{thm:msd_automaton}. 
\begin{proof}[of~\cref{thm:msd_automaton}]
Consider the partition $\cP$ constructed by combining the endpoints of the two partitions $\cP(-(c+1))$ and $\cP(-(c+2))$ from \cref{def:msd-partition}. Therefore, by~\cref{lem:ingrid} and~\cref{lemma:partition-intersection-property}, the partition $\cP$ is consistent with the msd-first Fibonacci representation. Among the intervals in $\cP$, let $J$ be the one that includes $0$. We define the DFAO $M_c=(Q,\Sigma,\delta,q_0,\Delta,\tau)$ as follows: 
\begin{itemize} \setlength{\itemsep}{0em}
    \item $Q = \cP$, \; $\Sigma = \{\tO,\tL\}$, \; $\Delta = \{0,1\}$, \; $q_0=J$,
    \item $\delta(I, a) = I_a$ where $I_a \in \cP$ is the interval that $\{\ndelta \phi\}$ falls in for every $n$ with $\{n\phi\} \in I$, 
    \item $\tau(I) = 
        \begin{cases}
            1, & \text{if $I \subseteq (-(c+1)\phi, -(c+2)\phi)\bmod 1$}; \\
            0, & \text{otherwise.}
        \end{cases}$. 
\end{itemize}
The correctness of the $\delta$ function is due to the fact that $\cP$ is consistent with the msd-first Fibonacci representation. The correctness of the $\tau$ function is due to~\cref{lemma:f(i+c)-ialpha} and the fact that by the construction provided in~\cref{def:msd-partition}, both $\{-(c+1)\alpha\}$ and $\{-(c+2)\alpha\}$ are endpoints of $\cP$.

Thus, based on the construction of $\cP(-(c+1))$ and $\cP(-(c+2))$, there are $O(\log c)$ endpoints in $\cP$, and therefore $O(\log c)$ intervals in $\cP$. This proves that there are $O(\log c)$ states in $M_c$.
\end{proof}

\begin{example}
    As an example, consider $c = 10$.  The automaton depicted in~\cref{fig:msd-automaton-c=10} illustrates our construction and the corresponding intervals. Since $-12=-13+1=-F_7+F_2$, the endpoints in the partition $\cP(-12)$ correspond to the integers $\{-12,-8,-5,-3,-2,-1\}$; and similarly, since $-11=-13+2=-F_7+F_3$, the partition $\cP(-11)$ corresponds to the integers $\{-11,-7,-5,-3,-2,-1\}$. Thus, we obtain the following partition $\cP$ of eight intervals:
\begin{align*}
\cP=\big\{ &[-5\phi, -8\phi)\bmod 1,\ [-8\phi,-3\phi)\bmod 1,\  [-3\phi,-11\phi)\bmod 1,\ [-11\phi,-\phi)\bmod 1  \\ &[-\phi,-12\phi)\bmod 1,\ [-12\phi,-7\phi)\bmod 1,\ [-7\phi,-2\phi)\bmod 1,\ [-2\phi,-5\phi)\bmod 1  \big\},
\end{align*}
and the construction in the proof of~\cref{thm:msd_automaton} provides the automaton in~\cref{fig:msd-automaton-c=10}.

\begin{figure}[!ht]
\centering
\scalebox{1}{
\begin{tikzpicture}
\tikzstyle{every node}=[shape=rectangle,fill=none,draw=black,minimum size=30pt,inner sep=2pt]

\node(58) at (0,0) {{\scriptsize{$[-5\phi,-8\phi)\bmod 1$}}$/0$};

\node(127) at (0,-2) {{\scriptsize{$[-12\phi,-7\phi)\bmod 1$}}$/0$};

\node(111) at (2,-4) {{\scriptsize{$[-11\phi,-\phi)\bmod 1$}}$/1$};

\node(25) at (4,-2) {{\scriptsize{$[-2\phi,-5\phi)\bmod 1$}}$/0$};

\node(112) at (10,-4) {{\scriptsize{$[-\phi,-12\phi)\bmod 1$}}$/1$};

\node(83) at (4,1.5) {{\scriptsize{$[-8\phi,-3\phi)\bmod 1$}}$/0$};

\node(72) at (6,0) {{\scriptsize{$[-7\phi,-2\phi)\bmod 1$}}$/0$};

\node(311) at (8,-2) {{\scriptsize{$[-3\phi,-11\phi)\bmod 1$}}$/0$};

\tikzstyle{every node}=[shape=circle,fill=none,draw=none,minimum size=10pt,inner sep=2pt]
\node(a0) at (-2,0) {};

\tikzstyle{every path}=[color =black, line width = 0.5 pt]
\tikzstyle{every node}=[shape=circle,minimum size=5pt,inner sep=2pt]
\draw [->] (a0) to [] node [] {}  (58);

\draw [->] (58) to [loop above] node [above] {$\tO$}  ();
\draw [->] (58) to [] node [left] {$\tL$}  (127);

\draw [->] (127) to [] node [left] {$\tO$}  (111);

\draw [->] (111) to [] node [above left] {$\tO$}  (25);
\draw [->] (111) to [bend left=5] node [above] {$\tL$}  (112);

\draw [->] (25) to [] node [left] {$\tO$}  (83);
\draw [->] (25) to [] node [left] {$\tL$}  (72);

\draw [->] (112) to [bend left=5] node [below] {$\tO$}  (111);

\draw [->] (83) to [] node [above] {$\tO$}  (58);
\draw [->] (83) to [bend left=50] node [right] {$\tL$}  (112);

\draw [->] (72) to [] node [right] {$\tO$}  (311);

\draw [->] (311) to [] node [below] {$\tO$}  (25);
\draw [->] (311) to [] node [right] {$\tL$}  (112);
;
\end{tikzpicture}
}
\caption{Msd-first DFAO for $(f(i+10))_{i \geq 0}$.}
\label{fig:msd-automaton-c=10}
\end{figure}
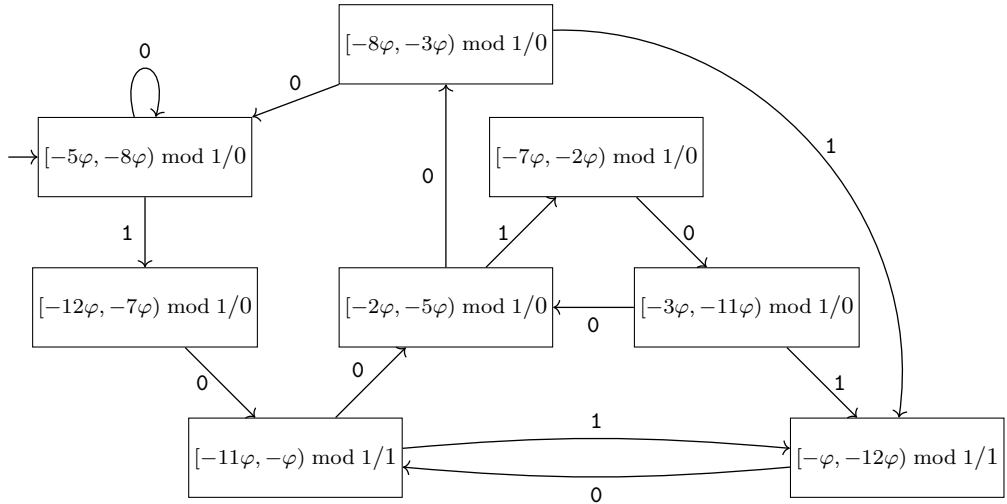
\end{example}

\section{Final words}

Examining the state complexity of transformations of Fibonacci-automatic sequences poses challenging questions at the border of formal language theory, combinatorics on words, and number theory.  There are many other questions we can examine.

It has not escaped our attention that there are obvious generalizations of these results, in some cases to quadratic irrationals, and in other cases to irrational numbers more generally.  Similarly, one can ask the same kinds of questions about the Tribonacci word (fixed point of the morphism $0 \rightarrow 01$, $1 \rightarrow 02$, $2 \rightarrow 0$), where the analysis is more puzzling.  We hope to take up these questions in a future paper.

\newpage
\appendix

\section{Lsd-first details} \label{appendix:lsd}

\subsection{Diophantine approximation and proof of~\cref{lem:lsd-partition}}

We first recall some basic properties of the Fibonacci numbers and the golden ratio. From the Binet formula for the Fibonacci numbers, one can deduce (or see, e.g., \cite{Knott:2021} for a reference)
that
\begin{equation}\label{eq:FnPhi_fract}
    \fract{F_i \phi}
    = \fract{(-1)^i \phi^{-i}}.
\end{equation}
Let $d\geq 2$. For every integer $n = \sum_{2\leq i \leq t} a_i F_i$, recall that $n^{[\leq d]} = \sum_{i=2}^{\min\{d,t\}} a_i F_i$, and $n^{[\geq d]} = \sum_{i=\max\{2,d\}}^{t} a_i F_i$. Since $n = n^{[\leq d]} + n^{[\geq d+1]}$, and from Eq.~\eqref{eq:FnPhi_fract}, we obtain 
\begin{equation}\label{eq:nphi_Ostr_small-large}
    \fract{n\phi}
    = \left(n^{[\leq d]} \phi
        + \sum_{i=d+1}^{t} (-1)^i a_i  \phi^{-i} \right) \bmod 1.
\end{equation}
We want to estimate the term $\sum_{i=d+1}^{t} (-1)^i a_i \phi^{-i}$; it can be understood as an error term.
Note that we get negative terms for odd indices $i$, and positive terms for even indices $i$.
The following simple identity, and the bound that follows from it, together show that we can indeed estimate the error term in the obvious way.
We have
\begin{equation}\label{eq:phi_sum}
    \sum_{i = 0}^\infty \phi^{-(d+1+2i)}
    = \phi^{-d}
    = \fract{(-1)^d F_d\phi},
\end{equation}
and, in particular,
\begin{align}\label{eq:phi_1-bound}
    \sum_{i = 0}^\infty \phi^{-(d+1+2i)}
    + \sum_{i = 0}^\infty \phi^{-(d+2+2i)}
    \leq \phi^{-2} + \phi^{-3} = \phi^{-1} < 1.
\end{align}
Thus, we get from Eqs.~\eqref{eq:nphi_Ostr_small-large}, \eqref{eq:phi_sum}, and \eqref{eq:phi_1-bound} that \begin{equation}\label{eq:large_intervals}
    \fract{n\phi} 
    \in \begin{cases}
        \left( {n^{[\leq d]}\phi} - \phi^{-d}, {n^{[\leq d]}\phi} + \phi^{-(d+1)} \right) \bmod 1,
        &\text{if $d$ is even};\\
        \left( {n^{[\leq d]}\phi} - \phi^{-(d+1)}, {n^{[\leq d]}\phi} + \phi^{-d} \right) \bmod 1,
        &\text{if $d$ is odd}.
    \end{cases}
\end{equation}
Moreover, if $a_d = 1$, then $a_{d+1} = 0$, so actually $\phi^{-(d+1)}$ cannot show up in the error term. Note that by the greedy property of the Zeckendorf representation, we have $a_d = 1$ if and only if $n^{[\leq d]}\geq F_d$. Thus, if $n^{[\leq d]}\geq F_d$, we have in fact
\begin{equation}\label{eq:small_intervals}
        \fract{n\phi} 
    \in \begin{cases}
        \left( (n^{[\leq d]} - F_{d+2})\phi, \, (n^{[\leq d]} - F_{d+1})\phi \right) \bmod 1,
        &\text{if $d$ is even};\\
        \left((n^{[\leq d]} - F_{d+1})\phi, \, (n^{[\leq d]} -F_{d+2})\phi\right) \bmod 1,
        &\text{if $d$ is odd}.
    \end{cases}
\end{equation}

We are now ready to prove~\cref{lem:lsd-partition}.

\begin{proof}[of \cref{lem:lsd-partition}]
From \eqref{eq:large_intervals} and \eqref{eq:small_intervals} and the fact that the points $\fract{n\phi}$ are dense in $(0,1)$, it is clear that the half open intervals $I_d(m)$, for $0\leq m \leq F_{d+1}-1$, cover $[0,1)$.

It remains to show that the intervals are disjoint. 
We do this simply by computing their total length.
By construction, we have exactly $F_d$ intervals of length 
\[
    \fract{(-1)^d(-F_{d+1} + F_{d})\phi} 
    = \fract{(-1)^{d-1}F_{d-1}\phi} 
    = \phi^{-(d-1)},
\]
and $F_{d+1} - F_{d} = F_{d-1}$ intervals of length 
\[
    \fract{(-1)^d(-F_{d+1} + F_{d+2})\phi} 
    = \fract{(-1)^{d} F_{d} \phi }
    = \phi^{-d}.
\]
Thus, the total length equals $F_d \phi^{-(d-1)} + F_{d-1}\phi^{-d} = 1$, via the Binet formula (or see, e.g., \cite{Knott:2021}). Therefore, the intervals must be disjoint, and form a partition. Also, \eqref{eq:partition_char} clearly holds as well by the arguments above.
\end{proof}

\subsection{Proof of \cref{thm:lsd-first-optimality}}

\begin{proof}[of \cref{thm:lsd-first-optimality}]
Let $c\geq 5$. We show that the construction in the proof of \cref{theorem:sc-lsd} is almost optimal. That is, we show that the constructed automaton $M_c=(Q,\Sigma,\delta,q_0,\Delta,\tau)$ is, possibly after one modification, minimal, and we count the number of states in the set $Q$.
Recall that $Q$ contains \begin{itemize}
    \item the initial state $q_\varepsilon$,
    \item the four sink states $A_{\tO},A_{\tL}, R_{\tO},R_{\tL}$,
    \item a certain number of intermediate states $q_x$, where the strings $x$ correspond to certain intervals. 
\end{itemize}
To be more precise about the intermediate states, we have that the strings $x$ correspond to pairs of integers $(d,m)$ such that $\fract{-(c+1)\phi} \in I_d(m)$ or $\fract{-(c+2)\phi} \in I_d(m)$, and $m = [x^R]_F$, where $z^R$ is the reverse image of $z$, and $x$ might be padded with $\tO$'s so that $|x| = d-1$. We distinguish the two following cases.

\caseI{1} There exists some $k\geq 1$ such that $F_k < c+1<c+2 \leq F_{k+1}$. Then by \cref{rem:c_in_partition} we know that $\fract{-(c+1)\phi}, \fract{-(c+2)\phi}$ become endpoints of the partition $\{I_{d}(m)\colon 0\leq m < F_{d+1}\}$ exactly at stage $d = k$.
This means that we only add states at the stages $d=2,3,\ldots k-1$.
Recall that in the construction of the automaton, at the first stage ($d=2)$, i.e., after reading one digit, we add either one or two states. In fact, we add only one state if and only if $f(c-1)= f(c-2) = 0$. Since by definition $g(c) = 1$ in this case, and $g(c) = 0$ otherwise, we add exactly $2-g(c)$ states at stage $d = 2$.
At all other stages ($d = 3, \ldots, k-1$), we add two states. Thus, we overall get $5 + (2 - g(c)) + 2(k-3) = 2k + 1 - g(c)$ states.
Since $F_k \leq c < F_{k+1}$, we have $|(c)_F| = k-1$, and the number of states is indeed $2(|(c)_F|+1)+1-g(c)$, as desired. 

It is clear that all states are reachable. Therefore, to prove minimality, we only need to prove that no two states are equivalent.
That is, we need to show that for all states $q\neq q' \in Q$ there exist strings $u,u' \in \Sigma^*$ such that $\tau(\delta(q,u)) \neq \tau(\delta(q',u))$.

Let $m_1, m_2$ be the unique integers such that $\fract{-(c+1)\phi} \in I_{k-1}(m_1)$ and $\fract{-(c+2)\phi} \in I_{k-1}(m_2)$. These intervals are ``critical intervals'' in the sense that at the next stage they each split into two subintervals which are each fully contained in $[-\phi, -2\phi) \bmod 1$ or in $[-2\phi, -\phi)\bmod 1$. Let $x_1, x_2$ be the corresponding ``critical strings'' with $|x_1| = |x_2| = k-2$ and $[x_1^R]_F = m_1$, $[x_2^R]_F = m_2$. 
Then, by construction, the strings $x_1\tO$ and $x_1\tL$ must be valid inputs and must lead to distinct outputs, and the same holds for the strings $x_2\tO$ and $x_2\tL$.

We now prove that no two states are equivalent. First, it is clear that none of the sink states $A_{\tO},A_{\tL}, R_{\tO},R_{\tL}$ can be equivalent to any other state. 

Now let $q = q_x$ and $q' = q_{x'}$ be two intermediate states corresponding to the (possibly empty) strings $x \neq x'$. By construction of $Q$, the strings $x, x'$ are each prefixes of $x_1$ or $x_2$. 
Let us write $x y = x_i$ and $x'y' = x_j$ for some (possibly empty) words $y,y'$, and $i,j \in \{1,2\}$ (where we might have $i=j$).
Assume without loss of generality that $|x|\leq |x'|$, and so $|y|\geq |y'|$ as $|x_1|=|x_2|$. We now distinguish two cases, according to whether $y=y'$ or not. 

\caseII{a} Assume that $y \neq y'$. Then we know that after reading $x' y$ we reach a sink state. Let $b$ be the output of that sink state, i.e., $\tau(\delta(q_\eps, x'y)) = b$ and for every digit $a\in \{\tO,\tL\}$, we have $\tau(\delta(q_\eps, x'ya)) = b$.
Since $x_i=xy$ is a critical string, we can choose a digit $c \in \{\tO,\tL\}$ such that $\tau(\delta(q_\eps, xyc)) = 1- b$. Thus, $\tau(\delta(q_{x'},yc)) = b$ and $ \tau(\delta(q_{x},yc))=1-b$, and so the two states  $q = q_x$ and $q' = q_{x'}$ are not equivalent.

\caseII{b} Assume that $y = y'$, i.e., $x_i = xy$ and $x_j = x'y$. Since $x\neq x'$, we have $i\neq j$, which means that we have states corresponding to prefixes of the two distinct critical strings $x_1,x_2$. 
By Eq.~\eqref{eq:FnPhi_fract}, the parity of $|x_1|= |x_2|$ determines whether reading $\tO$ after having read $x_1$ or $x_2$ means going to the subinterval to the left of the corresponding critical point, or to the right. 
However, for one of the two critical points, going to the left means that we output $0$, whereas for the other critical point, going to the left means that we output $1$. Therefore, $\tau(\delta(q_x,y\tO) \neq \tau(\delta(q_{x'},y\tO))$, and so the two states $q = q_x$ and $q' = q_{x'}$ are not equivalent. 

\caseI{2} There exists $k\geq 1$ such that $c+1 = F_{k+1}$ and $c+2 = F_{k+1} + 1$. In this case, by \cref{rem:c_in_partition}, the point $\fract{-(c+1)\phi}$ becomes an endpoint of the partition at stage $k$, whereas the point $\fract{-(c+2)\phi}$ becomes an endpoint of the partition at stage $k+1$. As in Case~1, let $x_1, x_2$ be the two critical strings, where now $|x_2| = |x_1|+1$.

Now, since $|x_1|,|x_2|$ have distinct parities, we see from Eq.~\eqref{eq:FnPhi_fract} that appending $\tO$ to $x_1$ means going to the corresponding left subinterval if and only if appending $\tO$ to $x_2$ means going to the corresponding right subinterval. Thus, appending $\tO$ either means output $0$ in both cases, or output $1$ in both cases.
Therefore, we can merge the last two intermediate states $q_{x_1}$ and $q_{x_2}$, since they are equivalent. 
Thus the total number of states remains $2(|(c)_F|+1)+1-g(c)$, as desired. 

The proof that no other two states are equivalent in this case is analogous to the proof in Case 1.
The only difference is that the argument in Case b of the minimality proof does not work. However, Case b never occurs here. Indeed, one can check that if $c = F_{k+1} - 1$, then the critical strings must be of the shapes $\tO^*$ and $(\varepsilon|\tO)(\tL\tO)^*$, and so, in particular, the second-to-last digits in $x_1,x_2$ are distinct.

\bigskip

Overall, we have proven that the number of states in the minimal lsd-first DFAO generating $(f(i+c))_{i \geq 0}$ is indeed $2(|(c)_F|+1)+1-g(c)$. Since both functions $c\mapsto |(c)_F|$ and $c\mapsto g(c)$ are Fibonacci-regular, the function that maps $c$ to the state complexity of $(f(i+c))_{i\geq 0}$ is also Fibonacci-regular; see \cite{Shallit:2023}, for instance, for more details on regular sequences. 
\end{proof}

\begin{example}
As an additional example, let $c=12$; we have $f(10)=f(11)=0$. Since $(-13)_F=(-\tL)\tO\tO\tO\tO$, $m_1=0$ and $(-14)_F=(-\tL)\tO\tO\tL\tO\tL\tO$, $m_2=7$. So in this case, we have two paths that are given by $\tO\tO\tO\tO$ and $\tO\tL\tO\tL\tO$, and the states $q_{\tO\tO\tO\tO}$ and $q_{\tO\tL\tO\tL\tO}$ are equivalent. Since $f(c+m_1)=f(c+m_2)=1$, this gives the minimal automaton presented in Figure~\ref{fig:ff12}.

\begin{figure}[H]
\centering
\includegraphics[width=6.5in]{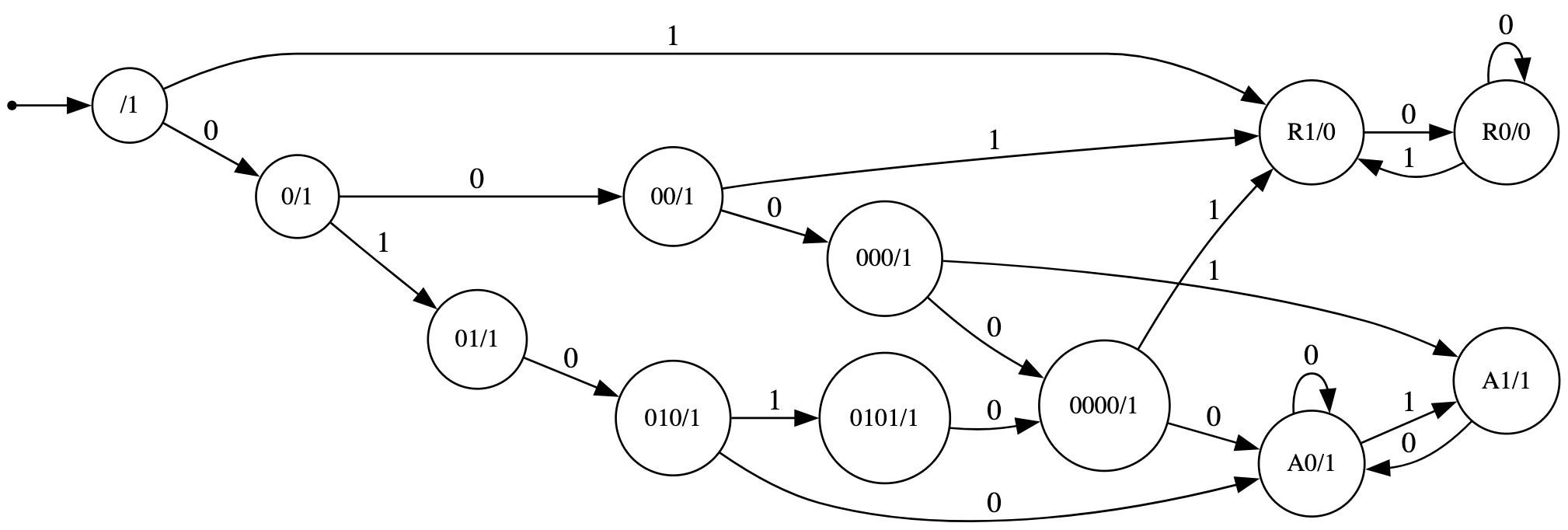}
\caption{Lsd-first DFAO for $(f(i+12))_{i \geq 0}$.}
\label{fig:ff12}
\end{figure}

\end{example}

\section{Msd-first details} \label{appendix:msd}

In order to prove \cref{lem:ingrid}, we need some auxiliary results. The following lemma states the expected property that dropping and appending a bit are inverse operations. 

\begin{lemma}\label{lem:append-drop}
For all integers $n$ with $n\geq 0$ or $n\leq - 3$ we have 
\begin{equation}\label{eq:append-drop}
    (\ndelta)' = n,
\end{equation}
where $a=0$ if $(n)_F$ ends with a $0$, and $a \in \{0,1\}$ otherwise. Moreover, for $n = -1$, Eq.~\eqref{eq:append-drop} holds for both $a=0$ and $a=1$, and for $n=-2$ Eq.~\eqref{eq:append-drop} holds only for $a=0$.
\end{lemma}

\begin{proof}
If $n\geq 0$ or $n\leq -3$, then the representation of $\ndelta$ in \eqref{eq:append_delta_pos} or \eqref{eq:append_delta_neg} is a valid representations (with the restriction that $a=0$ if $(n)_F$ ends with a $0$).
Therefore, dropping the last digit clearly reverses the operation.

Now let us check $n = -1$ and $n = -2$. If we append a $0$, then clearly we again get valid representations and dropping the $0$ reverses the operation.
For the case $n = -1 = -F_2$ and $a = 1$, we have
\begin{align*}
    \append{(-1)}{1}
    = -F_3 + F_2
    = -1,
\end{align*}
and since $(-1)' = -1$, Eq.~\eqref{eq:append-drop} holds in this case as well. 
For the case $n = -2 = -F_3$ and $a= 1$, we have $\append{(-2)}{1}= -F_4 + F_2 = -F_3= -2$, 
and so $(\append{(-2)}{1})' = -1 \neq -2$.
\end{proof}

\subsection{Flipping lemmas}

In this section we prove some fundamental but somewhat technical results about how our intervals change as we read additional bits of the input.  Our proofs are based on {\tt Walnut}, an automaton-based prover for properties of automatic sequences, including the Fibonacci word $\bf f$. For an introduction to {\tt Walnut}, see \cite{Shallit:2023}.

\begin{lemma}\label{lem:flip_append}
Let $x,y$ be negative integers such that the interval $I = (x \phi,y \phi) \bmod 1$ does not contain the two points $\fract{-\phi}$ and $\fract{-2\phi}$.
Let $n$ be a nonnegative integer. If $f(n) = 1$ set $a=0$, otherwise $a \in \{0,1\}$.
Then we have
\begin{equation}\label{eq:flip_append}
    \fract{n\phi} \in I
    \implies
    \fract{\ndelta\phi} \in (\append{y}{a} \phi, \append{x}{a} \phi)\bmod 1
    =: I_a.
\end{equation}
Moreover, if $x\neq -1$, the interval $I_a$ does not contain $\fract{-\phi}, \fract{-2\phi}$.
\end{lemma}

\begin{proof}
First we define the following {\tt Walnut} commands
\begin{verbatim}
reg lshift {0,1} {0,1} "([0,0]|[0,1][1,1]*[1,0])*":
# accepts (x,y) iff y is the left-shift of x.

reg rshift {0,1} {0,1} "([0,0]|[1,0][1,1]*[0,1])*(()|[1,0][1,1]*)":
# accepts (x,y) if y is the right-shift of x

def phin "?msd_fib (x=0 & y=0) | Ez $lshift(x-1,z) & y=z+1":
# (x,y) is accepted iff floor(x*phi) = y.

reg isfib msd_fib "0*10*":
# accepts x iff x is a positive Fibonacci number.

def smallest_fib_geq "$isfib(y) & y>=x & Az ($isfib(z) & z>=x) => z>=y":
# accepts (x,y) iff y is the smallest Fibonacci number >= x, for x>=1.
# 5 states

def fib_neg_repr "?msd_fib (x=0 & y=0 & z=0) | 
    (x>=1 & $smallest_fib_geq(x,y) & z+x=y)":
# accepts (x,y,z) iff -y+z is the negative Fibonacci expansion of -x,
    where y is a Fibonacci number.
# 9 states

def prime_inv_0 "?msd_fib Er,s,t,u $fib_neg_repr(x,r,s) &
    $rshift(t,r) & $rshift(u,s) & F[t]=@0 & F[u]=@0 & t=y+u":
# accepts (x,y) if -y is the integer obtained from 
    appending 0 to the representation of -x
# 7 states

def prime_inv_1 "?msd_fib (x=1 & y=1) | (x=2 & y=2) |
    Er,s,t,u $fib_neg_repr(x,r,s) &
    $rshift(t,r) & $rshift(u,s) & F[t]=@0 & F[u]=@1 & t=y+u":
# accepts (x,y) if -y is the integer obtained from 
    appending 1 to the representation of -x
# 9 states
\end{verbatim}

Note that the special cases $(x,y)=(1,1)$ and $(x,y) = (2,2)$ in {\tt prime\_inv\_1} account for the
fact that if we append $1$ to the representation of $-1, -2$, we get invalid representations. But these cases are still covered in the lemma.

Next, we define the two following functions about fractional parts. Notice that we are not using circular interpretation here. 

\begin{verbatim}
def compare_frac_part "?msd_fib 
    (x>y & Et,u,v $phin(x,t) & $phin(y,u) & $phin(x-y,v) & v+u<t) | 
    (x<y & Et,u,v $phin(x,t) & $phin(y,u) & $phin(y-x,v) & u<t+1+v)":
# accepts (x,y) iff {x*phi} < {y*phi}  for x, y >= 0 
# 22 states

def compare_frac_part2 "?msd_fib Et,u,v 
    $phin(x,t) & $phin(y,u) & $phin(x+y,v) & v>u+t":
# accepts (x,y) iff {-x*phi} < {y*phi} for x, y >= 1  
# 26 states
\end{verbatim}

\begin{verbatim}
def lies_in "?msd_fib  
 ($compare_frac_part(z,x) & $compare_frac_part2(x,y) 
             & ~$compare_frac_part2(y,z)) |
 ($compare_frac_part(x,z) & 
    ($compare_frac_part2(x,y) | ~$compare_frac_part2(z,y)))":
# checks if  {y*phi} is in the interval ({-x*phi}, {-z*phi}) 
    for the circular interpretation of intervals
# 97 states

def cw "?msd_fib 
    ($compare_frac_part(x,z) & $compare_frac_part(x,y) & 
        $compare_frac_part(y,z)) |
    ($compare_frac_part(z,x) & ($compare_frac_part(x,y) |   
        $compare_frac_part(y,z)))":
# checks if {y*phi} is in the interval ({x*phi}, {z*phi}) 
    for circular interpretation of intervals; cw = clockwise
# 44 states
\end{verbatim}

Now that we have these, we can rewrite the theorem statement in first-order logic and ask {\tt Walnut} to prove it:
\begin{verbatim}
eval check_flip_append_0 "?msd_fib An,n0,a,b,a0,b0 
    (a>=1 & b>=1 & F[n0]=@0 & 
    $rshift(n0,n) & $prime_inv_0(a,a0) & $prime_inv_0(b,b0) & 
    ($cw(a,1,b)|a=1|b=1) & ($cw(a,2,b)|a=2|b=2) 
    & $lies_in(a,n,b)) 
    => ($lies_in(b0,n0,a0))":
\end{verbatim}

\begin{verbatim}
eval check_flip_append_0_interval "?msd_fib An,n0,a,b,a0,b0 
    (a>=2 & b>=1 & F[n0]=@0 & 
    $rshift(n0,n) & $prime_inv_0(a,a0) & $prime_inv_0(b,b0) & 
    ($cw(a,1,b)|a=1|b=1) & ($cw(a,2,b)|a=2|b=2) & 
    $lies_in(a,n,b)) 
    => (($cw(b0,1,a0)|a0=1|b0=1) & ($cw(b0,2,a0)|b0=2|a0=2) )":
\end{verbatim}

\begin{verbatim}
eval check_flip_append_1 "?msd_fib An,n1,a,b,a1,b1 
    (a>=1 & b>=1 & F[n1]=@1 & 
    $rshift(n1,n) & $prime_inv_1(a,a1) & $prime_inv_1(b,b1) & 
    ($cw(a,1,b)|a=1|b=1) & ($cw(a,2,b)|a=2|b=2) & 
    $lies_in(a,n,b)) 
    => ($lies_in(b1,n1,a1))":
\end{verbatim}

\begin{verbatim}
eval check_flip_append_1_interval "?msd_fib An,n1,a,b,a1,b1 
    (a>=2 & b>=1 & F[n1]=@1 & 
    $rshift(n1,n) & $prime_inv_1(a,a1) & $prime_inv_1(b,b1) & 
    ($cw(a,1,b)|a=1|b=1) & ($cw(a,2,b)|a=2|b=2) & 
    $lies_in(a,n,b)) 
    => (($cw(b1,1,a1)|a1=1|b1=1) & ($cw(b1,2,a1)|b1=2|a1=2))":
\end{verbatim}
All of these return {\tt TRUE}.
\end{proof}

\begin{remark}
If $x=-1$ in~\cref{lem:flip_append}, then the interval $I_a$ does contains one of the two points $\fract{-\phi}, \fract{-2\phi}$.
Indeed, for $a = 0$, note that $\append{(-1)}{0} = -2$, so $I_0$ is of the form $I_0 =  (\append{y}{0} \phi,-2\phi) \bmod 1$ for every $y$. 
On the one hand, by the first part of~\cref{lem:flip_append}, we know that for all $n$ with $\fract{n\phi}\in I$ we have $\fract{\nzero\phi}\in I_0$. On the other hand, since the least significant digit of $\nzero$ is of course $0$, we know from~\cref{lem:01-areas} that $\fract{\nzero\phi}\in (-2\phi, -\phi)\bmod 1$. Thus, we have 
\[
    I_0 \cap (-2\phi, -\phi)\bmod 1=(\append{y}{0} \phi,-2\phi)\bmod 1 \cap (-2\phi, -\phi)\bmod 1 \neq \emptyset,
\]
which implies $\fract{-\phi} \in I_0$. Similarly, one can see that $\fract{-2\phi} \in I_1$.
\end{remark}

\begin{lemma}\label{lem:drop_flipinterval}
Let $x,y$ be two negative integers with $(x,y)\neq (-1,-2)$ and such that $\fract{-\phi} , \fract{-2\phi} \notin (x\phi,y\phi)\bmod 1$.
Let $r$ be a negative integer. Then
\[
    \fract{r\phi} \in  (x\phi,y\phi)\bmod 1
    \implies
    \fract{r'\phi} \in (y'\phi,x'\phi)\bmod 1.
\]
\end{lemma}
\begin{proof}
Once again, we use the {\tt Walnut} theorem-prover.
\begin{verbatim}
def negprime "?msd_fib (x=1 & y=1) | 
    (x>=2 & Er,s,t,u $ingrid_exp(x,r,s) & 
    $rshift(r,t) & $rshift(s,u) & y+u=t)":
\end{verbatim}

\begin{verbatim}
eval check_flip_drop "?msd_fib Aa,b,r,t,u,v 
    (a>=1 & r>=1 & b>=1 & (~((a=1&b=2)|(a=2&b=1))) & 
    $cw(b,r,a) & ($cw(a,1,b)|a=1|b=1) & ($cw(a,2,b)|a=2|b=2) &
    $negprime(a,t) & $negprime(b,u) & $negprime(r,v)) 
    => $cw(t,v,u)":
\end{verbatim}
And {\tt Walnut} returns {\tt TRUE}.
\end{proof}

\subsection{Proof of~\cref{lem:ingrid}}

We are now ready to prove~\cref{lem:ingrid}.

\begin{proof}[of~\cref{lem:ingrid}] Let $r = - F_R + \sum_{i=2}^{R-3} b_i F_i$ be the Fibonacci representation of $r\leq -2$, and define $r_k= -F_{R-k} + \sum_{i=2}^{R-3-k} b_{i+k} F_{i}$ for $0\leq k \leq R-2$. 
Recall that $\cP(r)$ is the partition induced by the points $r_k$ with $0\leq k \leq R-2$, and that $r_{R-2} = -1$ and $r_{R-3} = -2$. 
Thus, in particular, no interval $I \in \cP(r)$ contains $\fract{-\phi}$ or $\fract{-2\phi}$ in its interior. Let us prove that $\cP(r)$ is consistent with the msd-representation, i.e., that for every $I\in \cP(r)$ there exist $I_0, I_1 \in \cP(r)$ such that for all nonnegative integers $n$ we have
\[
    \fract{n \phi} \in I 
    \implies 
    \fract{\nzero \phi} \in I_0 \text{ and } \fract{\none \phi} \in I_1,
\]
where if $f(n) = 1$ we ignore $\none$ and $I_1$.
Note that all intervals in $\cP(r)$ are half-open, and all endpoints are of the shape $\fract{m\phi}$ with negative integers $m$.
Since for every nonnegative integer $n$ the point $\fract{n\phi}$ cannot coincide with an endpoint $\fract{m\phi}$ with negative $m$, we write all intervals from $\cP(r)$ as open intervals, for simplicity.

First consider the case $r=-2$. Then the partition $\cP(r) = \cP(2)$ consists of the two intervals $(-2\phi,-\phi)\bmod 1$ and $(-\phi, -2\phi)\bmod 1$. Let $I_0$ be the first interval, and $I_1$ the second one. Then, by~\cref{lem:01-areas}, we have $\fract{\nzero\phi} \in I_0$ and $\fract{\none \phi} \in I_1$ for all nonnegative $n$ (where we ignore $\none$ and $I_1$ if $f(n) = 1$).

Now assume $-r\leq -3$. In particular, we have that $\fract{-3\phi}$ is an endpoint in the partition, since $r_{R-3}=-3$. 
Let $I \in \cP(r)$ with 
\[ I = (r_k \phi, r_\ell \phi) \bmod 1,
\]
where $0 \leq k,\ell \leq R-2$ and $k\neq \ell$.

We distinguish the  following three cases; namely, if $r_k\leq -3$, $r_k=-1$, or $r_k=-2$. Indeed, the definitions of the desired sets $I_0,I_1$ will be slightly different in each case.

\caseI{1} $r_k\leq -3$. By~\cref{lem:flip_append} we have
\[
    \fract{n\phi} \in I \implies
    \fract{\ndelta\phi} \in 
        \left(\append{(r_\ell)}{a} \phi, 
        \append{(r_k)}{a} \phi\right) \bmod 1
        =: I_a
\]
for $a \in \{0,1\}$ (and $a = 0$ if $f(n) = 1$).
Moreover, the interval $I_a$ does not contain the two points $\fract{-\phi}, \fract{-2\phi}$.
Our goal is to prove that $I_a$ is fully contained in an interval from $\cP(r)$. 
We distinguish two subcases depending on the value of $r_\ell$. 

\caseII{1a} $r_\ell=-1$ or $r_\ell \leq -3$. 
To get a contradiction, assume that $I_a$ is not contained in an interval from $\cP(r)$. This means that there exists a point $\fract{r_j\phi}$, for some $0\leq j \leq R-2$ in the interior of $I_a$, i.e., 
\begin{equation}\label{eq:contrary_rj}
    \fract{r_j\phi} \in (\append{(r_\ell)}{a} \phi, \append{(r_k)}{a} \phi) \bmod 1.
\end{equation}
Since by assumption $r_k,r_\ell \neq -2$, we can now apply~\cref{lem:drop_flipinterval} and~\cref{lem:append-drop} to \eqref{eq:contrary_rj}, 
and we obtain
\begin{equation}\label{eq:rj_I}
    \fract{r_j'\phi} 
    \in \left(
        (\append{(r_k)}{a})'\phi,
        (\append{(r_\ell)}{a})' \phi
        \right) \bmod 1
    =(r_k \phi, r_\ell \phi) \bmod 1 
    = I.
\end{equation}
If $j < R-2$, we have $r_j' = r_{j+1}$ by definition. If $j = R-2$, we have $r_j' = (-1)' = -1 = r_j$.
Either way, Eq.~\eqref{eq:rj_I} contradicts the fact that $I \in \cP(r)$.

\caseII{1b} $r_\ell=-2$. In this case, $I=(r_k\phi,-2\phi)\bmod 1\subset (-\phi,-2\phi)\bmod 1$, so by \cref{lem:01-areas} we have $f(n)=1$ for all $n$ with $\fract{n\phi} \in I$. Thus, it suffices to consider $a=0$. Then the rest of the proof goes through exactly as in the previous subcase, since for $a=0$ we still have $(\append{(r_k)}{a})' = r_k$ and $(\append{(-2)}{a})' = -2$.

\caseI{2} $r_k=-1$. Then $I=(-\phi,r_\ell\phi) \bmod 1\subseteq (-\phi,-2\phi)\bmod 1$, and so as in Case 1b, we can set $a=0$. 
As before, \cref{lem:flip_append} implies that for $n$ with $\fract{n\phi} \in I$ we have
\[
    \fract{\nzero\phi} \in (\append{(r_\ell)}{0} \phi, \append{(-1)}{0} \phi) \bmod 1= (\append{(r_\ell)}{0} \phi, -2 \phi)\bmod 1.
\]
In view of~\cref{lem:01-areas}, since the least significant digit of $\nzero$ is of course $0$, we must actually have
\[
    \fract{\nzero\phi} \in \left(\append{(r_\ell)}{0} \phi, -\phi\right) \bmod 1=: I_0.
\] 
Now as in Case 1, we check that no other endpoint lies in $I_0$.
Assuming the contrary, we have $\fract{r_j \phi} \in (\append{(r_\ell)}{0} \phi, - \phi) \bmod 1,$  for some $0\leq j\leq R-2$.
Note that $\append{(r_\ell)}{0} \neq -2$ because $r_\ell \neq -1$.
Therefore, we can apply~\cref{lem:drop_flipinterval}, and obtain
\[
    \fract{r_j'\phi} \in (-\phi, r_\ell \phi) \bmod 1= I,
\]
and we can finish the proof as in Case 1a.

\caseI{3} $r_k=-2$, and so $I = (-2 \phi, r_\ell \phi) \bmod 1$. Note that $r_\ell = -1$ is impossible in this case since $r\leq -3$ and $\fract{-3\phi} \in (-2\phi, -\phi)\bmod 1$, so this last interval is not in $\cP(r)$.
Thus, we may assume $r_\ell \leq -3$. 
For $a=0$, \cref{lem:flip_append} again implies that for $n$ with $\fract{n\phi} \in I$ we have
\[
    (\append{(r_\ell)}{0} \phi, \append{(-2)}{0} \phi)\bmod 1
    =(\append{(r_\ell)}{0} \phi, -3\phi)\bmod 1
    := I_0,
\]
and the interval $I_0$ does not contain $\fract{-\phi}, \fract{-2\phi}$. We can show that $I_0 \in \cP(r)$ as in Case~1a (since for $a=0$ we still have $(\append{(-2)}{a})' = -2$).

Finally, for $a = 1$, \cref{lem:flip_append} implies that for $n$ with $\fract{n\phi} \in I$ we have
\[
    \fract{\none \phi}
    (\append{(r_\ell)}{1} \phi, -2\phi) \bmod 1
    :=I_1,
\] 
and the interval $I_1$ does not contain $\fract{-\phi}, \fract{-2\phi}$.
We check that no other endpoint lies in $I_1$.
Assuming the contrary, we have $\fract{r_j\phi} \in I_1$ for some $0\leq j\leq R-2$. Then from \cref{lem:drop_flipinterval} (note that $\append{(r_\ell)}{1}\neq -1$) and~\cref{lem:append-drop} (note that $r_\ell \neq -2$) we get that $\fract{r_j'\phi} \in (-\phi, r_\ell\phi) \bmod 1$ (since $(-2)'=-1$). 
Since $\fract{r_j\phi} \in I_1 \subset (-\phi,-2\phi)\bmod 1$, the least significant digit of $r_j$ is $1$. But then the second least significant digit has to be $0$. In other words, the least significant digit of $r_j'$ is $0$, and so in fact
\[
    \fract{r_j'\phi} \in (-2\phi, r_\ell\phi)\bmod 1=I,
\]
and we can finish the proof as in Case 1a. 
\end{proof}

\end{document}